\documentclass{article}

    \usepackage[final]{neurips_2024}

\usepackage{graphicx}
\usepackage{amsmath, amsthm, amssymb}
\usepackage{dsfont} %
\usepackage{url}
\usepackage{enumitem}
\usepackage{wrapfig}
\usepackage{pifont} %

\usepackage{multirow}
\usepackage{array}
\usepackage{tabularx}
\usepackage{tablefootnote}
\usepackage{colortbl} %
\usepackage{makecell}
\usepackage{float}

\usepackage[font=small, skip=8pt]{caption}
\usepackage{subcaption} %

\usepackage{bm}
\usepackage{threeparttable}
\usepackage{etoolbox}
\usepackage{xspace}
\usepackage{xparse}
\usepackage{setspace}
\usepackage[scientific-notation=true]{siunitx} %
\usepackage{algorithm}
\usepackage{algpseudocode}
\usepackage[normalem]{ulem} %

\usepackage[utf8]{inputenc}     %
\usepackage[T1]{fontenc}        %
\usepackage{hyperref}           %
\usepackage{url}                %
\usepackage{booktabs}           %
\usepackage{amsfonts}           %
\usepackage{nicefrac}           %
\usepackage{microtype}          %
\usepackage[dvipsnames]{xcolor} %
\usepackage{tcolorbox}          %

\title{
\resizebox{1.0\linewidth}{!}{Watermarking Makes Language Models Radioactive}
}
\author{
   Tom Sander$^*$ \\
   Meta FAIR \& École polytechnique  \\
  \And
  Pierre Fernandez$^*$\\ 
  Meta FAIR \& Inria Rennes
  \AND
  Alain Durmus \\
  École polytechnique
  \And
  Matthijs Douze \\
  Meta FAIR
  \And
  Teddy Furon \\
  Inria Rennes
}

\usepackage{annotate}

\usepackage{ifthen}  %
\newboolean{set_me_to_remove_todos}
\setboolean{set_me_to_remove_todos}{false} %
\ifthenelse{\boolean{set_me_to_remove_todos}}{
    \newcommand{\pierre}[1]{}
    \newcommand{\tom}[1]{}
    \newcommand{\todo}[1]{}
    \newcommand{\matthijs}[1]{}
    \newcommand{\teddy}[1]{}
    \newcommand{\alain}[1]{}
    
}{
    \newcommand{\tom}[1]{{\color{purple} [\textbf{Tom}: #1]}}
    \newcommand{\pierre}[1]{{\color{blue} [\textbf{Pierre}: #1]}}
    \newcommand{\matthijs}[1]{{\color{orange} [\textbf{Matthijs}: #1]}}
    \newcommand{\teddy}[1]{{\color{orange} [\textbf{Teddy}: #1]}}
    \newcommand{\alain}[1]{{\color{orange} [\textbf{AD}: #1]}}
    \newcommand{\todo}[1]{{\color{red} [\textbf{TODO}: #1]}}
    
}

\definecolor{metablue}{HTML}{0064E0}
\definecolor{metafg}{HTML}{1C2B33}
\definecolor{metabg}{HTML}{F1F4F7}

\newcommand{\cmark}{\ding{51}\xspace}%
\newcommand{\xmark}{\ding{55}\xspace}%
\newcommand{\xmarkg}{\textcolor{lightgray}{\ding{55}}\xspace}%

\newcommand{\R}{\mathds{R}}

\def\1{\mathbf{1}}

\def\V{\mathcal{V}}
\def\H{\mathcal{H}}
\def\Prob{\mathds{P}}

\newcommand{\eg}{e.g.,\@ }
\newcommand{\ie}{i.e.,\@ }
\newcommand{\aka}{a.k.a.,\@ }

\def\A{\mathcal A}
\def\B{\mathcal B}
\def\D{\mathcal D}
\newcommand{\logit}{\boldsymbol{\ell}}
\newcommand{\logpval}{\log_{10}(p)}
\newcommand{\pval}{p\textrm{-value}}

\newcommand{\supsetting}{supervised}
\newcommand{\Supsetting}{Supervised}
\newcommand{\unsupsetting}{un\supsetting}
\newcommand{\Unsupsetting}{Un\supsetting}

\newcommand{\aux}[1]{{\small{\textcolor{gray}{$\pm$#1}}}}
\definecolor{Gray}{gray}{0.95}
\newcommand{\colorcell}{\cellcolor{Gray}}

\newlength\savewidth\newcommand\shline{\noalign{\global\savewidth\arrayrulewidth
  \global\arrayrulewidth 1pt}\hline\noalign{\global\arrayrulewidth\savewidth}}

\newtheorem{definition}{Definition}

\newtheorem{proposition}{Proposition}

\begin{document}

\def \thefootnote{*} \footnotetext{
Equal Contribution. Correspondence at \{tomsander,pfz\}@meta.com \vspace{-1.5em}
}

\maketitle

\begin{abstract}
We investigate the \emph{radioactivity} of text generated by large language models (LLM), \ie whether it is possible to detect that such synthetic input was used to train a subsequent LLM.
Current methods like membership inference or active IP protection either work only in settings where the suspected text is known or do not provide reliable statistical guarantees.
We discover that, on the contrary, it is possible to reliably determine if a language model was trained on synthetic data if that data is output by a watermarked LLM.
Our new methods, specialized for radioactivity, detects with a provable confidence weak residuals of the watermark signal in the fine-tuned LLM.
We link the radioactivity contamination level to the following properties: the watermark robustness, its proportion in the training set, and the fine-tuning process.
For instance, if the suspect model is open-weight, we demonstrate that training on watermarked instructions can be detected with high confidence ($p$-value $< 10^{-5}$) even when as little as $5\%$ of training text is watermarked.
Radioactivity detection code is available at \url{https://github.com/facebookresearch/radioactive-watermark}
\end{abstract}

\begin{figure*}[h]
    \centering
    \includegraphics[width=1.0\textwidth, clip, trim=0 0.5cm 0 0]{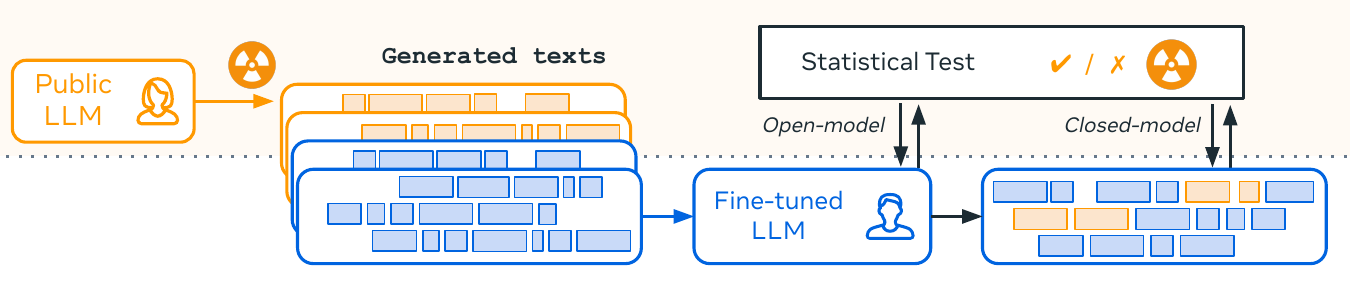}
    \captionsetup{font=small}
    \caption{
    Bob fine-tunes his LLM on data with a fraction coming from Alice's LLM.
    This leaves traces in Bob's model that Alice can detect reliably, provided that her text was watermarked.
    Thus, a side effect of Alice's watermark, intended for machine-generated text detection, is to reveal what data Bob's model was fine-tuned on.
    }
    \vspace{-0.2cm}
    \label{fig:fig1}
\end{figure*}

\section{Introduction}

Large Language Models (LLMs) are often instruction fine-tuned to align them with human prompts and improve their performance and generalization~\citep{ouyang2022training, wei2022finetuned, chung2022scaling}.
Fine-tuning requires expert knowledge to balance diversity and quality in the instruction dataset and a costly collection of manual annotations, especially for alignment~\citep{openai2023gpt, touvron2023llama2, team2023gemini}.
To address the cost and the difficulties of fine-tuning, practitioners often train on synthetic data generated by a model that has already been instructed, such as Bard, ChatGPT, or Claude. 
For example, works by \citet{wang2022self, honovich2022unnatural, peng2023instruction} created instruction data for many of the most recent LLMs~\citep{alpaca, xu2023baize, gunasekar2023textbooks, mukherjee2023orca}. 
This may also be unintentional when, for example, Turkers use ChatGPT to perform their tasks~\citep{veselovsky2023artificial}.
Such imitation raises questions about whether the fine-tuned model is a derivative work of the original model~\citep{wallace2020imitation}. 
In this context, it is crucial to understand how to detect when LLM outputs are used as training data.

Meanwhile, recent AI regulations enforce the transparency of generative models.
This is increasingly important in cases where the generated content may be used for malicious purposes~\citep{weidinger2022taxonomy, crothers2022machine}.
One approach is \emph{watermarking}. 
It embeds a secret trace in the synthetic content that can be detected to identify the generating model.
In the context of LLMs, recent techniques make detection efficient with minimal degradation of the generated text quality by altering the sampling of next tokens~\citep{aaronson2023watermarking,kirchenbauer2023watermark, kirchenbauer2023reliability}.

Based on these two observations, this study addresses the following question:
\begin{center}
\emph{What occurs when watermarked text is employed as fine-tuning data?}
\end{center}
We explore the potential ``radioactivity'' -- a term coined by~\citet{sablayrolles2020radioactive} -- of LLM watermarking,
which refers to the capacity of watermarked training data to contaminate a model.

We examine a model that has been fine-tuned on a corpus that may contain watermarked text (see Fig.~\ref{fig:fig1}).
The baseline method for detecting radioactivity executes the original watermark detection on the outputs generated by this model. 
However, this approach proves ineffective %
because the residual of the watermark is a weak signal hardly detectable in plain output text.
In this work, we are able to demonstrate that LLM watermarking is indeed radioactive thanks to our specific protocol designed for revealing weak contamination traces.
Our contributions include:

\begin{itemize}[leftmargin=*, itemsep=1pt, topsep=1pt]
    \item We design radioactivity detection methods for four scenarios based on model (\textit{open} / \textit{closed}) and training data (\textit{\supsetting} / \textit{\unsupsetting}) access. 
    Notably, our open-model detection (Fig.~\ref{fig:open_model}) improves the performance by orders of magnitudes.
\item 
    We show how to obtain reliable $p$-values for watermark detection when scoring millions of tokens. 
\item 
    We prove that watermarked text is radioactive in a real-world setting where an LLM is fine-tuned on slightly watermarked instruction data. 
    For instance, in the open-model scenario, our tests detect radioactivity with a $\pval$ of $10^{-5}$ when only 5\% of fine-tuning data is watermarked (Fig.~\ref{fig:wm-proportion}).
\end{itemize}

\vspace{-0.1cm}
\section{Background}

\subsection{Related work}\label{subsec:related}

\paragraph{Watermarking for LLMs.}
A recent branch of watermarking methods for decoder-only LLMs modifies either the probability distribution~\citep{kirchenbauer2023watermark} or the sampling method of the next token~\citep{aaronson2023watermarking, kuditipudi2023robust}.
Theoretical studies indicate that detectability depends on the entropy of generated text~\citep{christ2023undetectable, huang2023optimal}. 
Subsequent research suggests watermarking entropic passages, particularly in code~\citep{lee2023wrote}, while other works focus on ``semantic'' watermarks that depend on an entire past text's semantic representation~\citep{liu2023semantic, liu2024adaptive, fu2024watermarking}.

\citet{gu2023learnability} distill the methods used in this study within the model weights, allowing LLMs to generate watermarked logits natively, which is key for open-source models.
In contrast, we focus on unintentional contamination: 
Alice and Bob in Fig.~\ref{fig:fig1} are not collaborating, and Bob consumes only a small proportion of watermarked data.

\vspace{-0.2cm}
\paragraph{Membership inference attacks} 
(MIAs) aim to determine whether an arbitrary sample is included in a model's training data, with varying granularity on the adversary's knowledge~\citep{nasr2019comprehensive}.
Most of the time, the detection either build shadow models and observe a difference in their behavior~\citep{shokri2017membership, song2019auditing, hisamoto2020membership, mahloujifar2021membership} or directly observe the loss of the model~\citep{yeom2018privacy, sablayrolles2019white, watson2021importance, carlini2022membership}.
In the context of generative models, MIAs are intertwined with
\emph{dataset contamination} where one detects that an entire dataset
is part of the training data~\citep{shi2023detecting, golchin2023time}.
MIAs can violate the confidentiality of sensitive training or  reveal training on ``forbidden'' data, like
copyrighted material or evaluation data (which undermines benchmark results).

MIA may be used for radioactivity detection, but with a strong limitation.
Since it focuses on specific pieces of text, Alice in Fig.~\ref{fig:fig1} has to record all the outputs of her LLM.

\vspace{-0.2cm}
\paragraph{IP Protection.} Watermarking can be used for intellectual property protection. 
For instance, \citet{he2022cater,he2022protecting,li2023protecting} use lexical properties like synonyms whereas \citet{peng2023you} rely on backdoors. 
\citet{zhao2023protecting} develop a watermark dissuading model theft via distillation.
Yet its accuracy is empirical, it does not provide $p$-values that align with empirical false positive rates.

\paragraph{Radioactivity.}
\citet{sablayrolles2020radioactive} introduce the concept of \emph{radioactivity}: images are modified to leave a detectable trace in any classifier trained on them.
Our work studies the radioactivity of decoding-based LLM watermarks, which are primarily used to detect AI-generated text with proven accuracy.
We demonstrate that this form of radioactivity is reliably detectable across diverse settings.
Appendix~\ref{sec:ipp} details why existing IP protections or MIAs do not offer similar capabilities.

\subsection{Technical background for LLM watermarking}\label{sec:llm_watermarking}
\newcommand{\sk}{\mathsf{s}}

This paper focuses on watermarking schemes that modify the LLM decoding by hashing a watermark window.
In experiments we use~\citep{aaronson2023watermarking, kirchenbauer2023watermark} due to their omnipresence in the literature, their performance, and their practicality.
We briefly overview them and refer the reader to \autoref{app:watermarking} for details (on the watermark detection tests especially).

We consider a decoder-only LLM that takes as input a context (a sequence of tokens $\left( x^{(-C)}, ..., x^{(-1)} \right)\in \V^C$, $\V$ being the vocabulary of the model) and outputs a vector of logits $\logit\in \R^{|\V|}$.
Vector $\logit$ is transformed into $\mathbf{p}=\text{softmax}(\logit) \in [0,1]^{|\V|}$, the probability distribution of the next token.
The text is generated by sampling the next token $x^{(0)}$ from this distribution %
with some procedure (top-k sampling~\citep{fan2018hierarchical, radford2019language}, nucleus-sampling~\citep{holtzman2019curious}, etc.), then appending it to the context, and repeating the process.

The \textit{watermark embedding} alters the logit vector $\logit$ or the sampling procedure depending on a secret key.
Usually, the output of a secret-key cryptographic function hashes $k$ previous tokens $\left(x^{(-k)},\dots, x^{(-1)} \right)$ (the watermark window) and the secret-key $\sk$.
It serves as a seed for a random number generator, that influences the choice of the next token $x^{(0)}$.
For instance, \citet{kirchenbauer2023reliability} create a pseudo-random ``greenlist'' of tokens with proportion $\gamma$ of the entire vocabulary, whose logits are incremented by a quantity $\delta$, increasing their sampling probability.

The \textit{watermark detection} tokenizes a text, replays the seed generation and scores each token.
The score function on a current token $x^{(0)}$ may therefore be summed up as $W_{\textrm{score}}$ that also takes as input the watermark window $(x^{(-k)},\dots, x^{(-1)} )$, and  depends on the hashing's secret-key $\sk$:
\begin{figure}[h!]
   \vspace*{0.5em}
   \begin{equation}
   \label{eq:watermark_score}
   \eqnmarkbox[Plum]{token}{x^{(0)}} ;\, 
   \eqnmarkbox[Emerald]{window}{\big(x^{(-k)},\dots, x^{(-1)} \big)} 
   \mapsto 
   \eqnmarkbox[BurntOrange]{Wscore}{W_{\textrm{score}}} 
   \left(  
      \eqnmarkbox[Plum]{token2}{x^{(0)}}  ;\,  
      \sk, \eqnmarkbox[Emerald]{window2}{\big(x^{(-k)},\dots, x^{(-1)} \big)} 
   \right) \in \R.
   \end{equation}
   \annotate[yshift=-0.4em]{below,right}{token}{Current token being scored}
   \annotate[yshift=0.4em]{above,right}{window}{Watermark window ($k$ previous tokens)}
   \annotate[yshift=-0.4em]{below,right}{Wscore}{Scoring function (\eg $1$ if green token, $0$ otherwise)}
\end{figure} \\
A statistical test is performed on the cumulative score $S(X_N)$ -- which follows a known distribution in the absence of watermark, see App.~\ref{app:watermarking} -- where \( X_N := [x_1, \ldots, x_N] \)  is a list of $N$ tuples and:
\begin{equation}
\label{eq:def_S_N}
    S(X_N) := \sum_{t=1}^N W_\textrm{score} (    x^{(0)}_t ; \sk, (x^{(-i)}_t)_{i=k}^1 ).
\end{equation}

\vspace{-0.1cm}
\section{Problem Formulation}

\emph{Alice} owns a language model $\A$, fine-tuned for specific tasks such as chatting, problem solving, or code generation, which is available through an API (\autoref{fig:fig1}).
\emph{Bob} owns another language model $\B$.
Alice suspects that Bob fine-tuned $\B$ on some outputs from $\A$.
We denote by $D$ the dataset used to fine-tune $\B$, among which $D^{\A}\subset D$ is made of outputs from $\A$, in proportion $\rho = |D^{\A}|/|D|$.

\vspace{-0.2cm}
\paragraph*{Access to Bob's data.} 
\label{sec:degreeofsupervision}

We consider two settings for Alice's knowledge about Bob's training data:
\begin{itemize}[leftmargin=0.5cm, itemsep=2pt, topsep=1pt]

    \item \emph{\supsetting}: Bob queries $\A$ and  Alice retains all the content $\Tilde{D}^{\A}$ that $\A$ generated for Bob. Thus, Alice knows that $D^{\A} \subseteq \Tilde{D}^{\A}$. 
    We define the \emph{degree of supervision} $d := |D^{\A}|/|\Tilde{D}^{\A}|$,  

    \item \emph{\unsupsetting}: Bob does not use any identifiable account or is hiding behind others such that $|\Tilde{D}^{\A}| \gg |D^{\A}|$ and $d \approx 0$.
    This is the most realistic scenario.
\end{itemize}

Thus, $\rho$ is the proportion of Bob's fine-tuning data which originates from Alice's model
while $d$ quantifies Alice's knowledge regarding the dataset that Bob may have utilized (see Fig.~\ref{fig:datasets}).

\vspace{-0.2cm}
\paragraph*{Access to Bob's model.} 
We consider two scenarios:
\begin{itemize}[leftmargin=0.5cm, itemsep=2pt, topsep=1pt]
    \item Alice has an \emph{open-model} access to $\B$. 
    She can forward any inputs through $\B$ and observe the output logits.
    This is the case if Bob open-sources $\B$, or if Alice sought it via legitimate channels.
    \item Alice has a \emph{closed-model} access. 
    She can only query $\B$ through an API without logits access: Alice only observes the generated texts.
    This would be the case for most chatbots.
\end{itemize}

We then introduce two definitions of radioactivity:
\begin{definition}[Text Radioactivity]\label{def:text_radioactivity}
    Dataset $D$ is $\alpha$-radioactive for a statistical test $T$ if ``$\B$ was not trained on $D$'' $\subset \H_0$ and
    $T$ is able to reject $\H_0$ at a significance level ($p$-value) smaller than $\alpha$.
\end{definition}

\begin{definition}[Model Radioactivity]\label{def:model_radioactivity}
    Model $\A$ is $\alpha$-radioactive for a statistical test $T$ if
    ``$\B$ was not trained on outputs of $\A$'' $\subset \H_0$ and $T$ is able to reject $\H_0$ at a significance level smaller than $\alpha$.
\end{definition}

Thus, $\alpha$ quantifies the radioactivity of a dataset or model. 
A low $\alpha$, e.g. $10^{-6}$, indicates strong radioactivity: the probability of observing a result as extreme as the one observed, assuming that Bob's model was not trained on Alice's outputs, is 1 out of one million. 
Conversely, $\alpha\approx 0.5$ means that the observed result is equally likely under both the null and alternative (radioactive) hypotheses.

\newcommand{\amark}{\ensuremath{\sim}} %
\begin{figure}[b!]
   \begin{minipage}{0.62\textwidth}
        \centering
        \resizebox{1.0\linewidth}{!}{
        \begin{tabular}{r ccc ccc ccc}
            \toprule
            & \multicolumn{2}{c}{With WM} & \multicolumn{2}{c}{Without WM (MIA)} & \multicolumn{2}{c}{IPP} \\
            \cmidrule(lr){2-3} \cmidrule(lr){4-5} \cmidrule(lr){6-7}
            & Open & Closed & Open & Closed & Open & Closed \\
            \Supsetting\ & \cmark & \cmark &\cmark & \xmarkg & \cmark & \amark \\
            \Unsupsetting\ & \cmark & \cmark &\xmarkg & \xmarkg & \xmarkg & \xmarkg \\
            \bottomrule
        \end{tabular}
        }
        \captionof{table}{
            Availability of radioactivity detection under the different settings. 
            \textit{Open} / \textit{closed-model} refers to the availability of Bob's model, and \textit{\supsetting} / \textit{\unsupsetting} to Alice's knowledge of his data.
            Detection with watermarks is described in Sec.~\ref{sec:radioactivity_detection}, and a baseline without WM relying on MIA in App.~\ref{par:mia_wm}. 
            Intellectual Property Protection (IPP) refers to~\citet{zhao2023protecting}; see App.~\ref{sec:ipp}.    
        }
        \label{tab:summary_MIA_wm}
   \end{minipage}\hfill
    \begin{minipage}{0.33\textwidth}
        \centering
        \includegraphics[width=1.0\textwidth, clip, trim=0 2.2cm 0 0]{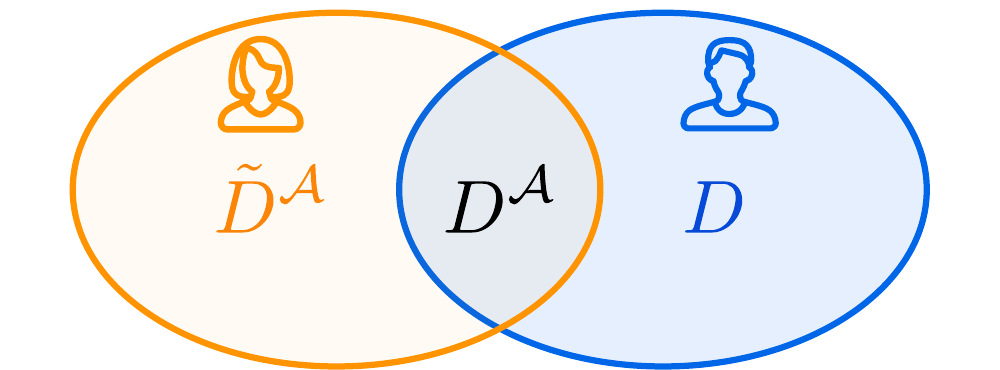}
        \captionsetup{font=small}
        \caption{
            Detection performance mainly depends on $\rho = |D^{\A}|/|D|$ and $d = |D^{\A}|/|\Tilde{D}^{\A}|$, where $D$ is the fine-tuning dataset used by Bob, $\Tilde{D}^{\A}$ are the outputs from Alice's model, and $D^{\A}$ the intersection of both.
        }
        \label{fig:datasets}
    \end{minipage}
\end{figure}

\vspace{-0.1cm}
\section{Radioactivity Detection}\label{sec:radioactivity_detection}

We build radioactivity detectors for all settings as shown in Tab.~\ref{tab:summary_MIA_wm}.
The outputs of $\A$ are watermarked with a method $W$ with Alice's secret key $\sk$ as described in Sec.~\ref{sec:llm_watermarking}. 
Note that, unlike watermark detection which takes text as input, the input for radioactivity detection is a model.

\subsection{Theoretical validity of statistical tests for radioactivity detection}\label{sec:statistical_test}

In the following, we construct a radioactivity test on $\B$ based on the detection score of Alice's watermarking method.
We focus on~\citet{kirchenbauer2023watermark} with the notations of Sec.~\ref{sec:llm_watermarking} for simplicity.
Each \(x_i\) in $X_N$ is a \((k+1)\)-tuple of tokens \((x_i^{(-k)}, \ldots, x_i^{(0)})\) and \(x_i^{(0)}\) is generated by $\B$ from the preceding tokens \( (x_i^{(-C_i)}, \ldots, x_i^{(-k)}, \ldots, x_i^{(-1)} ) \), where \( ( x_i^{(-C_i)}, \ldots, x_i^{(-k-1)}) := c_i\) is an optional context, e.g. prompt.
Note that the contexts \((c_1,\ldots,c_N)\) are crafted by Alice to better detect radioactivity in \((x_1,\ldots,x_N)\). As we see in Sec.~\ref{sec:detection-setup}, this can imply that they are watermarked themselves.
In this scenario, an accurate statistical test for detecting radioactivity can be performed through de-duplication.
Let \(\mathcal{H}_0\) be ``\(S(X_N)\) follows a binomial distribution \(B(N, \gamma)\)''.
Then:

\begin{proposition}\label{prop:test-validity}
``B was not trained on Alice’s watermarked data'' $ \subset \H_0$ \uline{if} tokens are \textit{de-duplicated}: (1) $(x_i)_{i \leq N}$ are pairwise distinct and (2) for each $1 \leq i \leq N$, $x_i$ is not in the context $c_i$.
\end{proposition}

The proof and a formal statement of this result are provided in App.~\ref{app:test-correctness-proof}.

With this inclusion, we can thus use a statistical test $T$ for hypothesis $\H_0$ to test against radioactivity. 
Specifically, assuming conditions (1) and (2) are met, the $p$-value $P(S(X_N)>t \mid \H_0)$ can be accurately computed using the regularized incomplete beta function $I_{\gamma}(t+1, N-t)$ detailed in App.~\ref{app:watermarking}.
The necessity of condition (1) to simulate i.i.d. scores is also used in classical watermark detection~\citep{fernandez2023three, kirchenbauer2023watermark}, but is more pronounced in our scenario, given the larger volume of tokens required to observe radioactivity.
(2) however is not necessary in classical watermark detection, as the prompt used to generate text is assumed not to be watermarked.
In our case,  we show in Sec.~\ref{sec:instruction} and App.~\ref{app:correctness} that both are critical in order to get reliable $p$-values.

\subsection{Radioactivity detection in practice}

\paragraph{Naive approach.} 
Text radioactivity can be detected by running $T$ on any
text generated by $\B$. 
However, this detector is weak because radioactivity can only be observed on an $x_i$ if it was part of $\A$’s watermarked outputs in $\B$’s training data.
This is because there is no correlation between the green lists of different watermark windows, since each partition is only a function of the watermark window through a hash function.

\paragraph{Overview.} 
We use the detection test $T$ detailed in Sec~\ref{sec:statistical_test} on $X_N$ generated through carefully crafted prompts, as our goal is to reduce noise by focusing on contexts likely to lead to radioactivity signals.
To this end, we recreate contexts similar to the ones that generated the watermarked text by Alice. 
We ensure the accuracy of statistical tests through de-duplication of scored tokens.

\begin{figure*}[b]
    \centering
    \vspace{-0.2cm}
    \includegraphics[width=1.0\textwidth, clip, trim=0 1.5cm 3.7cm 0]{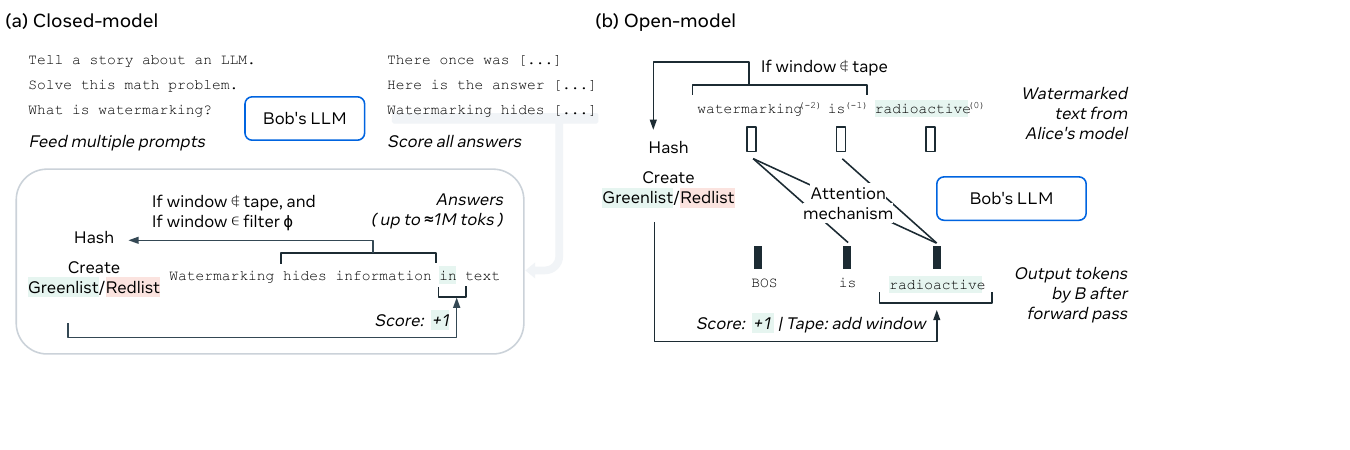}
    \caption{
    Radioactivity detection with closed or open model access (for simplicity, only \citep{kirchenbauer2023watermark} is illustrated).
    \textit{(Left)} New texts are generated from $\B$ using prompts from $\A$ and these texts are scored. 
    The filter $\phi$ is used to focus the score computation on likely contaminated $k$-grams. 
    \textit{(Right}) Text generated by $\A$ are directly forwarded through $\B$, and the next-token predictions are scored using tokens from the input as the watermark window.
    In both cases, the \textit{tape} ensures reliable $p$-values by de-duplicating scored tokens.
    }
    \label{fig:method}\label{fig:open_model}
\end{figure*}

\paragraph{Radioactivity detection in $\B$.}
To amplify the signals of radioactivity, we employ two strategies.
(1) In the \supsetting\ setting, we use watermarked text from $\Tilde{D}^\A$. In the \unsupsetting\ setting, we use other watermarked text generated by $\A$ that aligns with the suspected training distribution of $\B$ (e.g., English dialogues), but that was not used by $\B$.
(2) We score up to millions of tokens, orders of magnitudes more than usual.
The scoring depends on the access to $\B$:
\begin{itemize}[leftmargin=*, itemsep=1pt, topsep=2pt]
    \item \emph{closed-model}: we use the prompts to generate new texts from $\B$, and score these texts.
    \item \emph{open-model}, ``reading mode'': instead of generating completions with $\B$, we directly forward the text generated by $\A$ through $\B$, as depicted in Fig.~\ref{fig:open_model}.
    We then score next-token predictions with $W_{\textrm{score}}$ by using tokens from the input as watermark windows.
    Intuitively, for each green list token present in $\A$'s watermarked text, the context that led to its generation is reproduced. 
    This allows Alice to focus her analysis on how $\B$ responds to these specific contexts, which are more likely to lead to radioactivity signals than more generic ones.
\end{itemize}

\vspace{-0.2cm}
\paragraph{Filter on scored $k$-grams.}
To further improve detection in the closed-model setting where the reading mode is not possible, we only score $(k+1)$-tuples $x_i$ output by $\B$ for which the watermark window is often in $\A$’s watermarked outputs.
We thus introduce a filter $\phi$, a set that contains these watermark windows.
In the \emph{\supsetting} setting ($0<d\leq1$), $\phi$ is made of the $k$-grams present in $\Tilde{D}^\mathcal{A}$ (refer to Fig.~\ref{fig:datasets}).
In the \emph{\unsupsetting} setting, we focus on `likely' contaminated $k$-grams, \eg $k$-grams appearing in (new) watermarked text generated by $\A$.

\vspace{-0.2cm}
\paragraph{Token scoring and de-duplication.}
We score a token only if the same ($k+1$)-tuple $x_i$ has not been previously encountered.
Moreover, in the closed-model setting, we only score watermark windows ($k$-tuple) that are not part of the (watermarked) prompt. 
In the open-model setting, tokens with watermarked windows previously present in the attention span are not scored. 
This is achieved by maintaining a \emph{tape} memory of all such $k$-grams combinations during detection.
These adjustments ensure reliable $p$-values even when many tokens are analyzed.
This is empirically validated in Sec.~\ref{par:dedup-expe}, and additional details on the correctness of our tests are provided in App.~\ref{app:scoring}.

\vspace{-0.2cm}
\section{Radioactivity in Instruction Datasets}
\label{sec:instruction}
\vspace{-0.2cm}

This section considers a realistic scenario where a pre-trained LLM $\mathcal{B}$ is instruction fine-tuned on instruction/answer pairs generated by $\mathcal{A}$.
It shows that watermarked instructions are radioactive and compares the confidence of our different detection methods.

\subsection{Experimental setup of the instruction tuning}\label{sec:exp_setting}

\paragraph*{Instruction data generation.} 
We follow the Self-Instruct protocol~\citep{wang2022self} with $\mathcal{A}$=Llama-2-chat-7B~\citep{touvron2023llama2} (results hold for bigger teachers, see App.~\ref{app:bigger-teachers}).
We prompt the model with an instruction followed by three examples of instruction/answer pairs and ask it to generate the next $20$ instruction/answer pairs.
The sampling from the LLM logits is done with or without the watermarking method of~\citet{kirchenbauer2023watermark}, at logit bias $\delta=3.0$, proportion of greenlist tokens $\gamma=0.25$, and $k=2$. 
In both cases, we use nucleus sampling~\citep{holtzman2019curious} with $p=0.95$ and $T=0.8$.
We post-process the generated data to remove unfinished answers and near-duplicate instructions.
This yields a dataset of 100k instruction/answer pairs ($\approx$14M tokens).
Appendix~\ref{app:self_instruct} shows examples of these instructions and the  watermark detection rates.

Finally, we create six mixed datasets with $\rho$ \% of watermarked data (with $\rho \in \{ 0, 1, 5, 10, 50, 100\}$), filling the rest with non-watermarked instructions.
Thus, the total number of instructions is fixed at around 14M tokens, but the proportion of watermarked ones (which represents $\A$'s outputs) varies.

\vspace{-0.2cm}
\paragraph*{Fine-tuning.} 
We train $\mathcal{B}$ on these six datasets, closely following the approach of Alpaca~\citep{alpaca}:
we use AdamW~\citep{loshchilov2017decoupled} for 3000 steps, with a batch size of 8, a learning rate of $10^{-5}$ and a context size of 2048 tokens (which results in 3 training epochs).
The learning rate follows a cosine annealing schedule~\citep{loshchilov2017sgdr} with 100 warmup steps.
We set $\mathcal{B}$=Llama-1-7B~\citep{touvron2023llama}, a model trained on different datasets than $\mathcal{A}$=Llama-2, to avoid biases that could arise if the same base model were also used for fine-tuning.

\vspace{-0.2cm}
\subsection{Quality inspection of the instruction tuning}\label{sec:quality-inspection}

\begin{table}[t]
\centering
\begin{minipage}{0.57\textwidth}
    \caption{
    Evaluation of Llama-7B fine-tuned with varying proportions of watermarked instruction data.
    }\vspace{0.1cm}
    \label{tab:nlp_bench}\label{tab:perf}
    \footnotesize %
    \begin{tabular}{p{0.4cm}@{\hskip 16pt} | *{2}{p{0.5cm}} *{1}{p{0.7cm}} *{2}{p{0.5cm}} | p{0.6cm}}
        \toprule
        {} & \small{NQ} & \small{TQA} & \small{GSM8k} & \small{H.Eval} & \small{Avg.} & \small{MMLU}\\
        \midrule 
        \multicolumn{5}{l}{\small{\textit{Fine-tuned with $\rho$ \% of watermarked data:}}} \\
        $0\%$ & $5.0$ & $33.6$ & $11.8$ & $12.8$ & $15.8$ & $33.6$ \\
        $5\%$ & $5.2$ & $35.7$ & $11.2$ & $11.6$ & $15.9$ & $34.7$ \\
        $50\%$ & $4.1$ & $35.5$ &  $9.6$ & $12.8$ & $15.5$ & $35.0$ \\
        $100\%$ & $5.6$ & $36.4$ & $11.1$ &  $9.8$ & $15.7$ & $31.0$ \\
        \midrule
        Base & $3.2$ & $36.2$ & $10.5$ & $12.8$ & $15.7$ & $28.4$ \\
        \bottomrule
    \end{tabular}%
\end{minipage} \hfill
\begin{minipage}{0.38\textwidth}
    \centering
    \caption{
    Detection confidence $\logpval$ with varying supervision $d$ at $\rho=5\%$ of $\B$'s training data from $\A$, in the \textit{open}-model setting using the reading mode for detection (see Fig.~\ref{fig:open_model}).
    }\vspace{0.1cm}
    \label{tab:wm_supervision}
    \footnotesize %
        \begin{tabular}{l @{\hskip8pt} c}
            \toprule
            Supervision degree $d$ & $\logpval$ \\
            \midrule
            $0.1\%$ & $-5.8$\aux{1.8} \\
            $1\%$ & $-6.5$\aux{0.9} \\
            $5\%$ & $-16.0$\aux{2.6} \\
            $10\%$ & $<-30$ \\
            \bottomrule
    \end{tabular} 
\end{minipage}
\vspace{-0.3cm}
\end{table}

\begin{figure}[b]
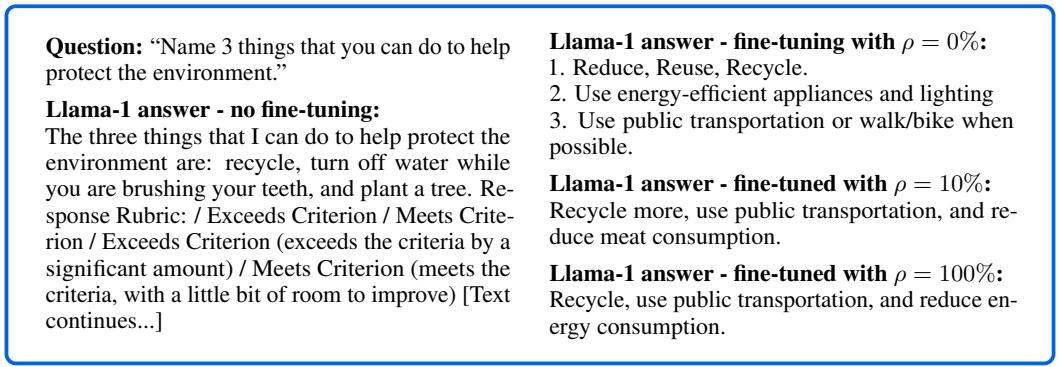

    \vspace{-0.4cm}
    \centering
    \begin{tcolorbox}[colframe=metablue, colback=white]
        \footnotesize
        \begin{minipage}{0.48\textwidth}
            \textbf{Question:} ``Name 3 things that you can do to help protect the environment.''
            \\[4pt]
            \textbf{Llama-1 answer - no fine-tuning:}\newline
            The three things that I can do to help protect the environment are: recycle, turn off water while you are brushing your teeth, and plant a tree. Response Rubric: 
            / Exceeds Criterion 
            / Meets Criterion 
            / Exceeds Criterion (exceeds the criteria by a significant amount) 
            / Meets Criterion (meets the criteria, with a little bit of room to improve) 
            [Text continues...]
        \end{minipage}\hspace{0.04\textwidth}%
        \begin{minipage}{0.48\textwidth}
            \textbf{Llama-1 answer - fine-tuning with $\rho=0\%$:}\newline
            1. Reduce, Reuse, Recycle. \newline 2. Use energy-efficient appliances and lighting \newline 3. Use public transportation or walk/bike when possible.
            \\[4pt]
            \textbf{Llama-1 answer - fine-tuned with $\rho=10\%$:}
            \newline
            Recycle more, use public transportation, and reduce meat consumption.
            \\[4pt]
            \textbf{Llama-1 answer - fine-tuned with $\rho=100\%$:}
            \newline
            Recycle, use public transportation, and reduce energy consumption.
        \end{minipage}
    \end{tcolorbox}
    \vspace{-0.2cm}
    \caption{
        Answers generated from Bob's model $\B$ (Llama-1), fine-tuned on instruction data generated by Alice's model $\A$ (Llama-2-chat) with different proportions $\rho$ of watermarked data.
        The quality of the instruction-tuning is not affected by the watermarking of the data (examples of training instruction/answer pairs are in Fig.~\ref{fig:example_self_instruct}).
    }
    \label{fig:example_answers_main}
\end{figure}

Alice's watermarking hyperparameters aim 
at 1) generating high-quality instructions and 2) ensuring that the watermark can be detected even in small text segments:
the watermark window size is $k=2$, sufficiently wide to eliminate biases yet narrow enough to make the watermark robust to edits;
$\delta=3$ yields high-quality text while ensuring that the watermark can be detected with a $p$-value of $10^{-6}$ on approximately 100 tokens (full results in App.~\ref{app:eval-wm-self-instruct}).

We inspect the outputs of the fine-tuned model $\mathcal{B}$ both qualitatively (see examples in Fig.~\ref{fig:example_answers_main} and App.~\ref{app:self_instruct}) and quantitatively in Tab.~\ref{tab:nlp_bench}.
We report 0-shot scores for an evaluation setup close to that of Llama: exact match score for Natural Questions~\citep{kwiatkowski2019natural} and TriviaQA~\citep{joshi2017triviaqa}; 0-shot exact match score without majority voting for GSM8k~\citep{cobbe2021training}; pass@1 for HumanEval~\citep{chen2021Evaluating}; and accuracy on MMLU~\citep{hendrycks2020measuring}.
As expected, instruction-tuning does not affect most benchmarks while enhancing it for MMLU, as in \citep{dettmers2023qlora}.
Thus, watermarking does not significantly impact the fine-tuned model performance.

\begin{figure}[b]
    \begin{minipage}{0.48\textwidth}
        \includegraphics[width=0.95\linewidth, clip, trim=0 0 0 0]{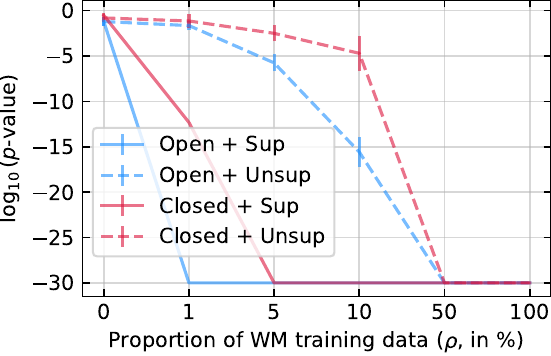}
        \caption{
            Radioactivity detection results.
            Average of $\logpval$ over 10 runs ($\downarrow$ is better). Bars indicate standard deviations.
            The detection methods are detailed in Sec.~\ref{sec:radioactivity_detection}.
            In the \supsetting\ closed-model setting, our tests detect radioactivity ($p<10^{-5}$) when only $1\%$ of training data is watermarked.
            In the absence of watermarked data, all tests output random $p$-values.
        }
        \label{fig:wm-proportion}
    \end{minipage} \hfill
    \begin{minipage}{0.48\textwidth}
    \centering
        \hspace{-0.2cm}
        \includegraphics[width=0.95\linewidth]{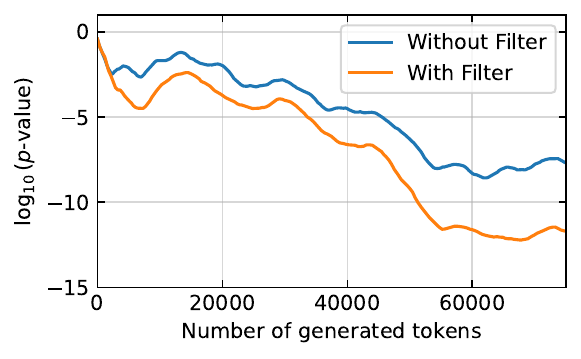}
        \caption{
            Influence of the filter on scored tokens.
            $\logpval$ as a function of the number of generated tokens in the \supsetting\ closed-model setting with $\rho = 1\%$. 
            We perform the watermark detection test on text generated by $\B$ with prompts from $\Tilde{D}^\A$. 
            When filtering, we only score $k$-grams that were part of $\Tilde{D}^\A$.
        }
        \label{fig:filter_non_filter_nbtok}
    \end{minipage}
\end{figure}

\subsection{Experimental setup of the detection}\label{sec:detection-setup}

In the \textbf{open-model} setting, we use a set of watermarked instructions generated by Alice's model $\A$ to score $N=225$k tokens.
For the \supsetting\ setting ($d=1$), we directly use all the $\rho\%$ watermarked texts among the 100k instructions used to train $\B$; for the \unsupsetting\ setting ($d=0$), we use watermarked instructions unused by $\B$.
We use the ``reading mode'' for detection (see Sec.~\ref{sec:radioactivity_detection}).

In the \textbf{closed-model} setting, $\B$ is only accessible via an API.
We prompt $\B$ with instructions generated by $\A$, concatenate all the answers, and score $N=600$k tokens, after filtering and de-duplicating the $k$-grams of $\approx1.5$M generated tokens.
This represents around $10^4$ queries if we assume an answer is $100$ tokens.
We score a token only if its previous $k$-gram is part of filter $\phi$. 
In the \supsetting\ setting ($d>0$), we collect the watermarked prompts/answers from $\Tilde{D}^\A$, part of which were used for fine-tuning, and define $\phi$ as the set of all the $k$-grams present in it.
In the \unsupsetting\ setting, we generate $100$k new watermarked instructions with $\A$ and save all $k$-grams into $\phi$.
In both cases, we run the detection $10$ times on different chunks of text and report averaged results.

\subsection{Detection results}

\paragraph{Proportion of watermarked data and access to Bob's model and data.}
\autoref{fig:wm-proportion} first presents the $p$-values of our tests for different proportions $\rho$, under the 4 possible scenarios, and shows that the detection confidence increases with the proportion of watermarked data, and with Alice's knowledge.

The \supsetting\ setting ($d=1$) is straightforward: radioactivity is detected with a $p$-value smaller than $10^{-30}$ (resp. $10^{-10}$) in the open-model (resp. closed-model) setting, even if only $1\%$ of Bob's fine-tuning data originated from $\A$.
Indeed, with open-model access, we only score 1) $k$-grams that can actually be contaminated and 2) within a context that matches the one seen during fine-tuning.
In the closed-model access, prompting $\B$ with its training questions enhances radioactivity detection, as its answers are likely to reproduce the watermarked responses from $\A$ used during training.

In the \unsupsetting\ setting ($d=0$), our open-model radioactivity detection test still yields $p<10^{-5}$ when no more than $5\%$ of the instructions used to fine-tune $\B$ originate from Alice's model.
The detection is done on a corpus of texts that does not contain samples seen by Bob at training time.
However, it contains watermark windows that likely overlap with Bob's training data, on which radioactivity may be detected.
In the closed-model setting, the detection is less powerful: it requires a higher proportion of watermarked data to be detected with a $p$-value smaller than $10^{-5}$.

\paragraph{Influence of the degree of supervision.}
To study the intermediate regime of weak supervision, we fix $\rho=5\%$ in the open-model detection. 
We vary the supervision degree $d$ by mixing the watermarked data used for fine-tuning with new watermarked data generated by $\A$ in a ratio $d$ to simulate $\Tilde{D}^\A$. 
We then run the radioactivity detection test with the reading mode on this corpus.
\autoref{tab:wm_supervision} shows that the detection confidence increases with the supervision degree (from $-5.8$ to $<-30$ when $d$ goes from $0.1\%$ to $100\%$).
Even for very weak supervision, the detection is still effective, with a $p$-value smaller than $10^{-5}$ when $d=0.1\%$.
This is in stark contrast with MIA-based methods for which supervision is necessary (we detail this in Fig.~\ref{fig:mia-comparative-analysis} of App.~\ref{par:mia_wm}).

\vspace{-0.2cm}
\paragraph{Influence of the filtering in the closed-model setting.}
\autoref{fig:filter_non_filter_nbtok} compares detection with and without the $\phi$ filter when $1\%$ of fine-tuning data is watermarked in the \supsetting\ closed-model setting (with $d=1$).
We plot the $\log_{10}(p\textrm{-value})$ against the number of generated tokens.
As expected, the detection confidence increases with the number of tokens.
Moreover, filtering consistently brings improvements: after scoring $75000$ tokens, the $\logpval$ equals $-12$ with filter and $-8$ without.
Filtering appears particularly important to increase the detection confidence on the worst-case scenarios (see the largest $\pval$ observed over the 10 runs in Fig.~\ref{fig:box_plot} of App.~\ref{appendix:experiments}).

\paragraph{Influence of the de-duplication on the correctness of radioactivity tests.}\label{par:dedup-expe}
For a statistical test to be valid, the $p$-value should be uniformly distributed between 0 and 1 under the null hypothesis $\H_0$ (mean $0.5$, standard deviation of $\approx0.28$).
Accurate theoretical $p$-values are particularly important in the common case where obtaining samples of fine-tuned $\B$ is expensive. 
This cost limits the sample size, reducing the power of empirical tests and compromising the confidence in their results. 

We validate our tests and highlight the importance of de-duplication by observing the empirical $p$-value distribution after scoring 1M tokens when $\B$ is not trained on watermarked data ($\H_0$). 
We first run 10 detection tests when scoring distinct $\{$watermarked window + current token$\}$ combinations, as suggested by~\citep{kirchenbauer2023watermark, fernandez2023three}. 
This approach is insufficient on its own, as shown in Tab.~\ref{tab:pval_h0} in the ``without de-duplication'' column.
For instance,  in the closed-model setting, the average $\pval$ is inferior to $10^{-30}$ even if the model does not output watermarked texts.
This is a strong false alarm, \ie incorrectly rejecting the null hypothesis.
The additional de-duplication rules (Sec.~\ref{sec:radioactivity_detection}) resolve this issue. 
We observe in the ``with de-duplication'' column that the empirical $p$-values match the theoretical ones, indicating that the test results are valid and reliable. 
These statistics, computed across various secret keys and seeds, are detailed in App.~\ref{app:repporting}. 
App.~\ref{app:correctness} gives additional details, supported by additionnal experiments, on the importance of de-duplication.

\subsection{Intermediate summary \& discussion}
Our watermark-based radioactivity detection methods can identify $10^{-5}$-radioactivity in model $\B$ across various setups. 
Even in the most realistic scenario -- \unsupsetting\ access to Bob's data -- our method holds true with only closed access to $\B$, given that at least $10\%$ of the data is watermarked. 
Open-model access further enhances the test's statistical significance, detecting radioactivity with a $p$-value smaller than $10^{-10}$ when $10\%$ of the fine-tuning data is watermarked (Fig.~\ref{fig:wm-proportion}).
Note that we deliberately score different numbers of tokens for the open and closed-model settings. 
It shows that the open-model is more efficient since it requires fewer tokens to achieve lower $p$-values (see Fig.~\ref{fig:wm-proportion}) . 
Furthermore, since $p$-values plateau beyond a certain threshold (see Fig.~\ref{fig:filter_non_filter_nbtok}), additional token scoring in the open-model setting does not significantly enhance detection.

\textbf{Other approaches.}
The applicability of other approaches is summarized in Tab.~\ref{tab:summary_MIA_wm}: ``\xmarkg'' means that no method in the literature currently tackles this problem with LLMs, and ``$\sim$'' means that methods that address the problem have strong technical issues (the statistical guarantees do not hold).
Indeed, we show in App.~\ref{par:mia_wm} that other passive methods like Membership Inference Attacks (MIAs) are effective only in the \textit{\supsetting} setting, where Alice has precise knowledge of the data used to train Bob's model and \textit{open} access to it. 
In that scenario, she can demonstrate $10^{-30}$-radioactivity, providing strong evidence that Bob has trained on her model. 
App.~\ref{sec:ipp} demonstrates that even state-of-the-art \emph{active} protection methods~\citep{zhao2023protecting} fall short in the \unsupsetting\ setting.

Besides, our tests provide reliable $p$-values thanks to meticulous de-duplication.
When $\B$ is not trained on watermarked data, the $p$-values are not overly low (Fig.~\ref{fig:wm-proportion}, Tab.~\ref{tab:pval_h0}, App.~\ref{app:scoring}). 
App.~\ref{sec:ipp} shows that this is not the case for recent synonym replacement-based watermarks~\citep{he2022protecting, he2022cater}.

\renewcommand{\aux}[1]{{\scriptsize \textcolor{gray}{$\pm$#1}}}
\begin{table}[t!]
    \begin{minipage}{0.44\linewidth}
        \centering
        \renewcommand{\arraystretch}{1.1}
        \caption{
        Average $p$-values under $\H_0$ ($\B$ \emph{not} trained on watermarked data: $p$ should be 0.5). 
        In the open-model setting (resp. closed), we exclude a token if the same watermark window is already present in the attention span (resp. in the watermarked prompt).
        Without de-duplication, $p$-values are overly low: the test does not work.
        }\label{tab:pval_h0}
        \footnotesize %
        \begin{tabular}{l @{\hskip8pt} c c}
            \toprule
            & \multicolumn{2}{c}{De-duplication} \\
            \cmidrule(rr){2-3}
            Access to Model & With & Without \\
            \midrule
            Open & 0.46\aux{0.27} & 0.053\aux{0.12} \\
            Closed & 0.42\aux{0.30} & $<10^{-30}$ \\
            \bottomrule
        \end{tabular}
    \end{minipage} \hfill
    \begin{minipage}{0.53\linewidth}
    \centering
        \caption{
            Influence of watermarking method and $k$ on radioactivity. Average $\log_{10}$ $p$-values.
            ``Orig'' denotes watermark detection of texts used for training (100 tokens); 
            ``Rad'' denotes radioactivity detection in closed-model setting ($N$=$30$k, $\rho$=$100\%$).
            Both KGW~\citep{kirchenbauer2023reliability} and AK~\citep{aaronson2023watermarking} behave the same way and lower $k$ increases radioactivity. 
        }\label{tab:exp_kgram}
        \footnotesize %
        \renewcommand{\arraystretch}{1.1}
        \begin{tabular}{r r c c c}
            \toprule
           \multicolumn{2}{c}{Window Size $k$} &  1 & 2 & 4 \\ %
           \hline 
           \multirow{2}{*}{KGW} & Orig & -8.6\aux{4.4} & -6.4\aux{3.9} & -6.7\aux{4.0} \\
                                & Rad & -43.7\aux{7.4} & -10.2\aux{3.0} & -1.4\aux{0.6} \\
           \hline
           \multirow{2}{*}{AK}  & Orig & -7.7\aux{4.5} & -7.6\aux{5.1} & -7.1\aux{5.1} \\
                                & Rad & -47.2\aux{4.5} & -18.4\aux{2.8} &  -2.8\aux{3.2} \\
            \bottomrule
        \end{tabular}
    \end{minipage}
    \vspace{-0.5cm}
\end{table}

\vspace{-0.1cm}
\section{Investigating Radioactivity}\label{sec:fine-tuning-abl}

This section further studies what influences radioactivity from three angles: 
fine-tuning, watermarking algorithm, and data distribution. 
We also explore how Bob can try to remove radioactivity traces.
Other scenarios such as multi-bit watermarking~\citep{yoo2024advancing} and other influencing factors, such as the size of model $\A$ or mixes of different fine-tuning data, are studied in App.~\ref{appendix:experiments}.

\begin{table}
	\centering
    \caption{
        Influence of the model fine-tuning on the radioactivity.
        We report the $\logpval$ for $10$k scored observations (lower means more radioactive).
        {\setlength{\fboxsep}{2pt}\colorbox[HTML]{F2F2F2}{Gray}} indicates values used in Sec.~\ref{sec:instruction}.
    }
    \vspace{-0.1cm}
    \label{tab:ft-abl}
    \renewcommand{\arraystretch}{1.2}
    {\small
  	\subfloat[ Learning rate.]{
            \centering
            \hspace{-0.03\linewidth}
            \begin{minipage}{0.23\linewidth}
                {\begin{center}
                    \vspace{-0.2cm}\begin{tabular}{ccc}
                       \colorcell  $10^{-5}$ & $5\cdot 10^{-5}$ & $10^{-4}$\\
                        \shline
                       \colorcell  -32.4 & -49.6 & -58.0\\
                    \end{tabular}
                \end{center}}
            \end{minipage} 
        } \hspace{0.05\linewidth}
  	\subfloat[ Epoch.]{
            \centering
            \begin{minipage}{0.27\linewidth}
                {\begin{center}
                    \vspace{-0.2cm}\begin{tabular}{ *{4}{c} }
                        1 & 2 & \colorcell 3 & 4 \\
                        \shline
                        -20.8 & -29.2 & \colorcell -33.2 & -34.8 \\
                    \end{tabular}
                \end{center}}
            \end{minipage} 
        } \hspace{0.05\linewidth}
        \subfloat[ Adapters.]{
            \centering
            \begin{minipage}{0.15\linewidth}
                {\begin{center}
                    \vspace{-0.2cm}\begin{tabular}{cc}
                        \colorcell Full & Q-LoRA \\
                        \shline
                        \colorcell -32.4 & -11.0 \\
                    \end{tabular}
                \end{center}}
            \end{minipage} 
        } \hspace{0.05\linewidth}
        \subfloat[ Model size.]{
            \centering
            \begin{minipage}{0.15\linewidth}
                {\begin{center}
                    \vspace{-0.2cm}\begin{tabular}{cc}
                       \colorcell  7B & 13B\\
                        \shline
                       \colorcell -32.4 & -33.2 \\
                    \end{tabular}
                \end{center}}
            \end{minipage} 
        } 
    }
    \vspace{-0.2cm}
 \end{table}

\vspace{-0.2cm}
\subsection{Fine-tuning}

We first study the influence of fine-tuning on the same setup as Sec.~\ref{sec:instruction}, with regards to: 
(a) the learning rate,
(b) the fine-tuning algorithm,
(c) the number of epochs,
(d) the model size.
We fine-tune $\B$ with the same dataset of $\rho=100\%$ watermarked instructions and the same parameters.
We detect radioactivity in the \emph{open-model} / \emph{\unsupsetting} setting.
This is done on $N=10$k next-predictions, and where the texts that are fed to $\B$ are watermarked instructions generated with $\A$.
\autoref{tab:ft-abl} reports the results. The more the model fits the data, the easier its radioactivity is to detect.
For instance, multiplying the learning rate by $10$ almost doubles the average $\logpval$ of the test.

\vspace{-0.2cm}
\subsection{Watermarking method \& data distribution}

To introduce more variety to the data under study, we now prompt $\A$=Llama-2-7B with the beginnings of Wikipedia articles in English and generate the next tokens with or without watermarking. 
We then fine-tune $\B$=Llama-1-7B on the natural prompts followed by the generated answers.
The fine-tuning is done in 1000 steps, using batches $8\times2048$ tokens (similarly to Sec.~\ref{sec:instruction}).
This section fine-tunes $\B$ on $\rho=100\%$ English watermarked texts.
We explore two aspects of the radioactivity.

\vspace{-0.2cm}
\paragraph{Watermark window size.} 
\autoref{tab:exp_kgram} highlights that the confidence of the detection decreases with $k$ when fixing the $p$-value of the watermark detection of the training texts.
There are two explanations. 
First, for lower $k$, the chances that a ($k+1$)-tuple repeats in the training data are higher, which increases its memorization.
Second, the number of watermark windows is $|\V|^k$ and therefore increases with $k$, while the number of watermarked tokens is fixed. 
Thus, at detection time, the proportion of radioactive tuples decreases with $k$, diminishing the test's power. 
This experiment also demonstrates that the methods of \cite{aaronson2023watermarking} and \cite{kirchenbauer2023reliability} behave similarly.

\begin{wraptable}{r}{0.6\linewidth}
    \centering
    \renewcommand{\arraystretch}{1.2}
    \vspace{-0.3cm}
    \caption{
    Influence of the target text distribution on detection.
    $\B$ is prompted with beginnings of Wikipedia articles in the corresponding language, and detection is done on generated next tokens. 
    For each language, we score $N=250$k $k$-grams using the \textit{closed-model} setting described in Sec.~\ref{sec:radioactivity_detection}.
    }
    \label{tab:exp_language}
    \footnotesize
    \vspace{-0.2cm}
        \begin{tabular}{ c @{\hskip16pt} *{5}{p{0.9cm}} }
            Language &  English & French & Spanish & German & Catalan  \\
            \shline
             $\logpval$ & $<$-50 & -7.8 &  -5.7 &  -4.0 &  -2.1 \\
        \end{tabular}
    \vspace{-0.3cm}
\end{wraptable}

\vspace{-0.2cm}
\paragraph{Data distribution.}\label{par:distrib}
We consider an \unsupsetting\ setting where Alice has no prior knowledge about $D^\A$, the data generated with $\A$ used to fine-tune $\B$.
As an example, Alice does not know the language of $D^\A$, which could be Italian, French, Chinese, etc. 
We run the detection on text generated by $\B$, with prompts from Wikipedia in different languages.
The confidence of the test on another language -- that might share very few $k$-grams with $D^\A$ -- can be low, as shown in Tab.~\ref{tab:exp_language}.

Alice may, however, combine the $p$-values of each test with Fisher's method.  
This discriminates against $\mathcal{H}_0$: ``\textit{none of the datasets are radioactive}'', under which the statement ``\textit{Bob did not use any outputs of $\A$}'' falls.
Therefore, the test aligns with our definition of model radioactivity as per definition~\ref{def:model_radioactivity}.
From Tab.~\ref{tab:exp_language}, Fisher's method gives a combined $p$-value of $<10^{-50}$. 
Thus, even if Alice is unaware of the specific data distribution generated by $\A$ that Bob may have used to train $\B$ (\eg problem-solving scenarios), she may still detect radioactivity by combining the significance levels.

\subsection{Possible defense: ``Purification''}\label{sec:possible-defenses}

\begin{wraptable}{r}{0.395\textwidth}
    \centering
    \vspace{-0.4cm}
    \caption{Sequential fine-tuning to remove watermark traces. The first fine-tuning is done with the setup presented in Sec.~\ref{sec:instruction}, with $\rho$=$10\%$ of watermarked data, and the second on OASST1. Detection is performed in the open/unsupervised setting.}
    \label{table:results_FineTuning}
    \footnotesize
    \begin{tabular}{c|c}
        \toprule
        $2$\textsuperscript{nd} fine-tuning & Average $\logpval$ \\
        \midrule
        \xmark & -15 \\
        \cmark & -8 \\
        \bottomrule
    \end{tabular}
    \vspace{-0.3cm}
\end{wraptable}
We investigate the impact of a second fine-tuning on human-generated data to remove the watermark traces.
After having trained his model on a mix of watermarked and non-watermarked data as in Sec.~\ref{sec:instruction}, Bob fine-tunes his model a second time on text from OASST1~\citep{kopf2024openassistant}, with the same fine-tuning setup.

\autoref{table:results_FineTuning} shows that the second-fine-tuning divides by $2$ the significance level of the statistical test, but does not remove all traces.
Other attempts to remove radioactivity could be to rephrase the watermarked instructions or use differentially private training.
The logic is overall the same as previously pointed out in the fine-tuning ablations:
if the original watermark is weaker or if the fine-tuning overfits less, then radioactivity will be weaker too.
Therefore, the radioactivity detection will be less powerful, but given a sufficient amount of data, Alice may still be able to detect it.

\vspace{-0.1cm}
\section{Conclusion}

This study formalizes the concept of ``radioactivity'' in language models.
It introduces methods to detect traces that LLM-generated texts leave when used as training data.
We show that this task is difficult for non-watermarked texts in the most realistic scenarios. 
Yet, watermarked texts exhibit significant radioactivity, contaminating models during fine-tuning.
This makes it possible to identify with high confidence if outputs from a watermarked model have been used to fine-tune another one (although it may not be used to detect the use of the other model itself).

\section*{Acknowledgments}
We thank Herv\'e J\'egou, Alexandre Sablayrolles, Pierre Stock, Chuan Guo and Saeed Mahloujifar for insightful comments and useful discussions.
Teddy Furon's work is supported by ANR / AID under Chaire SAIDA ANR-20-CHIA-0011.

\clearpage
\bibliographystyle{plainnat}
\bibliography{references}

\clearpage
\appendix
\section{Limitations}\label{sec:limitations}
\begin{itemize}[leftmargin=16pt, itemsep=1pt, topsep=1pt]
\item Our research specifically explores the radioactivity of LLM watermarks originally designed for detecting AI-generated text. 
Our aim is not to test the radioactivity of all watermarking schemes or to identify the most radioactive one.
Instead, we aim to develop a general method for detecting radioactivity, applicable to similar decoding-based watermarks.
Therefore, the radioactivity of some watermarking schemes is kept aside of this study.
We detail this point about the generalization of our work to other watermarking schemes in App.~\ref{app:generalization}.

\item We discuss various parameters that influence radioactivity and how it is linked to the robustness of the used watermark, in Sec.~\ref{sec:fine-tuning-abl} and App.~\ref{appendix:experiments}.
While we have tried to span as many factors as possible, we acknowledge that some are missing. 
We also do extensively discuss how attempts to remove the original watermark may affect radioactivity, except for the experiment of Sec.~\ref{sec:possible-defenses} where Bob re-fine-tunes his model.
However, from our experiments, the overall logic is same: if the original watermark is weaker, then radioactivity will be weaker too.

\item The computational efficiency of the proposed algorithms is linear in the number of tokens to analyze and depends on the suspected LLM. 
In the closed-model setting, we need to query the model about 1000 times to detect radioactivity confidently, which could be a limitation if the suspect model is only available through an API.

\item One objective of Sec.~\ref{sec:instruction} is to demonstrate that we can detect radioactivity with high confidence in a realistic setting. 
From this, we assert in the abstract that we can detect radioactivity with a $p$-value of $10^{-5}$, even when only as little as 5\% of the fine-tuning instruction/answer pairs are watermarked (open/\unsupsetting\ setting).
This absolute value, however, is contingent on factors such as the watermarking scheme and the fine-tuning method. 
We do not claim that 5\% is a universally sufficient watermarking proportion for all fine-tuning scenarios.
Nevertheless, the relative insights derived from the $p$-values in Sec.\ref{sec:instruction} will always hold true. 
For instance, that the open-model method outperforms the closed-model method.

\end{itemize}

\section{Broader Impacts}\label{app:broader-impacts}

\subsection{Societal impact}

The societal impact of this work includes positive and negative aspects. 
On the one hand, it can assist proprietary LLMs in preventing imitation and protect their intellectual property. 
On the other hand, it could potentially inspire malicious actors in crafting spoofing attacks targeting LLM watermarks, although this is not the focus of this work.

\subsection{Environmental impact}
The cost of the experiments and of model fine-tuning is high, though order of magnitude less than other LLM papers. 
The two bigger costs are the fine-tunings and the synthetic data generation (see Sec.~\ref{app:compute_ressources}).
We also roughly estimate that the total GPU-days used for running all our experiments to $1000$, or $\approx 25$k GPU-hours.
This amounts to total emissions in the order of $2$ tons of CO$_2$eq.
Estimations are conducted using the \href{https://mlco2.github.io/impact#compute}{Machine Learning Impact calculator} presented by~\citet{lacoste2019quantifying}.
We do not consider in this approximation: memory storage, CPU-hours, production cost of GPUs/ CPUs, etc.

\subsection{Compute resources}\label{app:compute_ressources}

For our experiments, we utilized an internal cluster.
\begin{itemize}
    \item The generation of both watermarked (100k instruction/answer pairs) and non-watermarked (100k instruction/answer pairs) instructions took approximately one day on a single node equipped with 8 V100 GPUs. 
    This process was repeated 5 additional times for the experiments on the watermark window size of \autoref{tab:exp_kgram} for both~\citet{aaronson2023watermarking} and~\citet{kirchenbauer2023reliability}.
    \item Each main experiment in Sec.~\ref{sec:instruction} involved the fine-tuning of a 7B model, which took approximately eight hours on a single node with 8 V100 GPUs.
    Radioactivity detection in the closed-model setting required prompting the model approximately $10,000$ times, which took around five hours on a single V100 GPU.
    \item Each subsequent radioactivity detection test took less than 30 minutes and did not require a GPU. 
    For the open-model setting, no additional generation was needed. We simply forwarded 1M tokens in batches, which took approximately two hours with a single GPU.
\end{itemize}

\section{Details on LLM Watermarking}\label{app:watermarking}

We use the notations presented in Sec.~\ref{sec:llm_watermarking}.
This paper considers the watermarking methods~\citep{kirchenbauer2023watermark, kirchenbauer2023reliability, aaronson2023watermarking} which alter the logit vector $\logit$ or the probability vector $\mathbf{p}$ when trying to generate $x^{(0)}$, depending on the window of $k$ previous tokens in the context: $x^{(-k)}, ..., x^{(-1)}$ ($x^{(i)}$ is an integer between 0 and the size of the vocabulary).

A hash function maps these tokens to a random seed.
The hash function also depends on a secret key $\sk$.
In our work, the hash function operates according to the equation: $$h_{n+1} = (h_{n} \cdot \sk + x^{(n)}) \mod (2^{64} - 1),$$ 
for $n \in -k, .., -1$, and $h_{-k} = 0$.
In this equation, $h_{n+1}$ is the updated hash value, $h_{n}$ is the current hash value, $\sk$ is a constant, $x^{(n)}$ is the current input token, $\mod$ is the modulus operator, and $2^{64} - 1$ is the modulus value. 
The final seed ($h_0$) is used to initialize a random number generator (RNG).
RNG is then used to influence or determine the next token's choice $x^{(0)}$.

\subsection{\citet{kirchenbauer2023reliability}} 
RNG is used to create a greenlist containing $\gamma |\V|$ tokens, where $\gamma \in [0,1]$.
The logit of every token in the greenlist is incremented by $\delta$.
The sampling then proceeds as usual.
Intuitively, this encourages the generation of greenlist tokens by increasing their probability.
\\ \noindent \\
\emph{For detection}, one tokenizes the text and counts how many tokens are in the greenlist of their window.
More formally, we consider a text of $N$ tokens. 
Recall that $W_\textrm{Score} ( x^{(0)}; \sk, (x^{(-i)})_{i=k}^1) = \mathds{1} (``x^{(0)} \text{ is in the greenlist of } ( \sk, (x^{(-i)})_{i=k}^1 )" )$.
A statistical test on the cumulative score is then performed based on $N$ tuples \( X_N := [x_1, \ldots, x_N] \), where each \(x_i\) is a \((k+1)\)-tuple of tokens \((x_i^{(-k)}, \ldots, x_i^{(0)})\) and
token \(x_i^{(0)}\) is generated from the preceding tokens \( (x_i^{(-C_i)}, \ldots, x_i^{(-k)}, \ldots, x_i^{(-1)} ) \). Each tuple $x_i$ also includes an optional context \( ( x_i^{(-C_i)}, \ldots, x_i^{(-k-1)}) := c_i\) (e.g., prompt). 
The cumulative score $S(X_N)$ associated with $X_N$ is the number of greenlist tokens, which in the absence of a watermark follows a binomial distribution:
\begin{equation}
    S(X_N) = \sum_{t=1}^N  W_\textrm{Score} (    x^{(0)}_t ; \sk, (x^{(-i)}_t)_{i=k}^1 )  \, .
\end{equation}

The statistical test considers two hypotheses, $\H_0$: ``\textit{the text is natural}'' against $\H_1$: ``\textit{the text was generated with watermark}''.
Under $\H_0$, we suppose that the $\{0,1\}$-valued random variables $\left(W_\textrm{Score} (    x^{(0)}_t ; \sk, (x^{(-i)}_t)_{i=k}^1 )\right)_t$ are independent and identically distributed as a Bernoulli distribution with parameter $\gamma$.
Therefore, $S$ follows a binomial distribution with parameters $N$ and $\gamma$.
The $p$-value of a test associated with score $s$, \ie probability of obtaining a score higher than $s$ under $\H_0$, can be obtained theoretically from:
\begin{equation}
    \text{$p$-value}(s) = \Prob(S(X_N) \geq s | \H_0) = I_{\gamma}(s+1,N-s),
\end{equation}
where $I_{\gamma}$ is the regularized incomplete Beta function.
Under $\H_1$, the score is likely to be higher than under $\H_0$, so the $p$-value is likely to be lower.

The \textit{strength} of the watermark is mainly controlled by the parameter $\delta$.
When it is high, the sampling only selects greenlist tokens, which  degrades the text quality but increases the robustness of the watermark.

\subsection{\citet{aaronson2023watermarking}}
RNG is used to generate a random vector $R \in [0,1]^{|\V|}$.
Then, instead of sampling from distribution $p$, the next token is chosen by $x^{(0)} = \arg \max_{v \in \V } R_v^{1/p_v}$ (nucleus sampling or top-$K$ can be applied to $p$ before computing $R^{1/p}$).
Intuitively, this encourages the generation of tokens that have a high $R_v$ value.
It also presents the interesting property that $\forall v\in \V$, $\Prob_R (x^{(0)}=v) = p_v$.
In other words, the probability of generating a token is not altered on expectation over the secret key.
\\ \noindent \\
\emph{For detection}, one goes through all tokens. 
At time-step $t$, the $k$ previous tokens are used to retrieve the key vector $R^{(t)} \in [0,1]^{|\V|}$.
We denote by $R_t$ the number $R^{(t)}_{x^{(t)}}$, \ie the value of the key vector for the token in the text. 
The score is now: 
\begin{equation}
S=-\sum_t \ln (1-R_t),
\end{equation}
and with the notations of Sec.~\ref{sec:llm_watermarking}: 
\begin{align*}
W_\textrm{score}  \left(x^{(t)} ; \sk, \big(x^{(t-i)}\big)_{i=k}^1 \right)  = - \ln (1-R_t) = - \ln \left(1- R^{(t)}_{x^{(t)}} \right).
\end{align*}

We consider the same hypothesis testing as before.
Under $\H_0$, we assume that $R_t \sim \mathcal{U}(0,1)$ and that $R_t$ are i.i.d., so $S$ follows a $\Gamma(N,1)$ distribution.
The $p$-value of a test associated with score $s$ reads:
\begin{equation}
    \text{$p$-value}(s) = \frac{\Gamma(N,s)}{\Gamma(N)},
\end{equation}
where $\Gamma$ is the upper incomplete gamma function.
Under $\H_1$, the score is expected to be higher.
In fact, its expectation is lower-bounded by $N + cH$, where $c$ is a positive constant and $H$ is the entropy of the generated text.

The \textit{strength} of the watermark is directly linked with the temperature $\theta$ of the softmax.
For instance, for very high values of $\theta$, the softmax outputs an almost uniform probability vector $p$, so the choice of the next token is determined entirely by $R$ (the token with highest $R$ value is chosen) -- whereas for very low $\theta$, distribution $p$ is very peaky so $R$ has little influence.

\section{Score Computation}\label{app:scoring}

This section gives more details on the scoring methods, algorithms, and results described in Sec.~\ref{sec:radioactivity_detection}.

\subsection{Reporting}\label{app:repporting}

\paragraph{Log $p$-values.}
Given that $p$-values often span various orders of magnitude, we consistently report the average of the $\logpval$ over multiple runs rather than the average of the $p$-values themselves. 
In the main text, we interpret the average $\logpval$ as though it could be directly read as a $p$-value (for instance, if the average $\logpval$ is $-5$, we interpret it as if Alice would be incorrect to reject the null hypothesis only once in 10,000 instances). 
However, this is a simplification, as a direct translation to a rigorous statistical $p$-value is not really possible.
Therefore, we show the box-plots with additional statistics in Fig.~\ref{fig:box_plot_open}.

\paragraph{Average over multiple runs.} 
Due to computational constraints, the standard deviations for the $\logpval$ are not calculated across multiple $\B$ models trained on different instruction data for each setting. 
Instead, for each setting, we generate the same volume of data (14M tokens, see~\autoref{sec:instruction}) in addition to the data used to fine-tune $\B$. 
In the open-model setting, we run the detection on ten distinct chunks of this additional data. 
In the closed-model/\unsupsetting\ setting, we prompt $\B$ with ten different chunks of new (non-watermarked) sentences and score the responses.
For the closed-model/fully \supsetting\ setting (results reported in Fig.~\ref{fig:wm-proportion}), we score the answers from $\B$ to all the prompts present in $D^\A$, which for $\rho=1\%$ of watermarked fine-tuning data only represent $75$k tokens. 
It explains the absence of confidence intervals.

\begin{figure}[b!]
    \centering
    \includegraphics[width=0.5\linewidth]{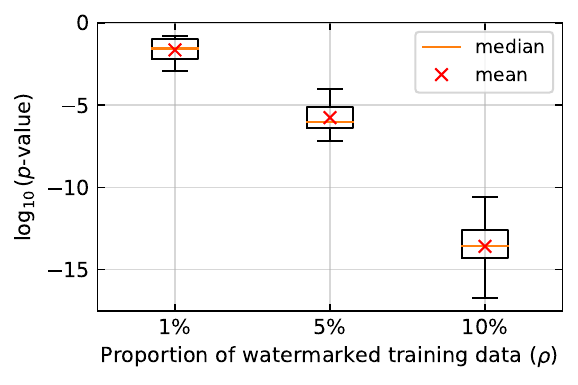}
    \captionsetup{font=small}
    \caption{Box plot for the $\logpval$ in the open/\unsupsetting\ setting with varying $\rho$, the proportion of $\B$'s fine-tuning data watermarked. 
    This corresponds to the values presented in Fig.~\ref{fig:wm-proportion} where the means are reported.}
    \label{fig:box_plot_open}
\end{figure}

\subsection{Correctness}\label{app:correctness}

We focus here on the correctness of the tests used in this work to detect radioactivity.
App.~\ref{app:test-correctness-proof} gives the proof of validity for the radioactivity detection test proposed in Sec.~\ref{sec:statistical_test} and surfaces another point of view for this test.
App.~\ref{app:dedup_details} gives more details on why de-duplication is important in various scenarios.
Finally App.\ref{app:correctness-expes} empirically shows that the detection tests provide reliable $p$-values.

\subsubsection{Proof of correctness for radioactivity detection}\label{app:test-correctness-proof}

We give the proof of Proposition~\ref{prop:test-validity} of Sec.~\ref{sec:statistical_test} for the watermark of \citet{kirchenbauer2023watermark}. 
In this context, we remind that \(\mathcal{H}_0\) is ``\(S(X_N)\) follows a binomial distribution \(B(N, \gamma)\)''.

\begin{proposition}
    ``B was not trained on Alice’s watermarked data'' $ \subset \H_0$ \uline{if} 
    $(x_i)_{i \leq N}$ are independent and independent of Alice's Watermarking process.
    In practice, this is achieved if tokens are \textit{de-duplicated}: (1) $(x_i)_{i \leq N}$ are pairwise distinct and (2) for each $i$, $x_i$ is not in the context $c_i$.
\end{proposition}

\begin{proof}
With (2), the de-duplication excludes the possibility of simply repeating watermarked tuples from the prompt. 
Moreover, assuming an ideal hashing function, there is no correlation between the green lists of different watermark windows, so this simple rule ensures that each \((k+1)\)-tuple has a probability \(\gamma\) of contributing to an increase in the score.
If (1) also holds, \(X_N\) is a set of \(N\) unique \((k+1)\)-tuples of tokens, which is used to simulate independent score increments in watermarking detection~\citep{kirchenbauer2023reliability, fernandez2023three}.
Therefore by definition of $S(X_N)$ in
\eqref{eq:def_S_N}, $\left(W_\textrm{Score} (    x^{(0)}_t ; \sk, (x^{(-i)}_t)_{i=k}^1 )\right)_t$ are i.i.d. simulations distributed according to a Bernoulli distribution with parameter $\gamma$.
Thus, \(S(X_N)\) follows a binomial distribution \(B(N, \gamma)\).
\end{proof}

Note that the exact same arguments are valide for other hashing-based watermarks.
\paragraph{An alternative choice of $\H_0$ for radioactivity detection.}
Sec.~\ref{sec:statistical_test} adopts a definition of $\H_0$, which we show is adequate to radioactivity detection if de-duplication is carefully done.
Another approach closer to classical watermarking is also possible.
Specifically, \citet{kirchenbauer2023watermark} define $\H_0$ as ``The text sequence is generated with no knowledge of the red list rule.''.
If we keep the same $\H_0$ for radioactivity, the goal of de-duplication is to modify the score computation such that we know the distribution of this score under $\H_0$.
Filtering/de-duplication does not change the null hypothesis. 
It only modifies the observation from which we compute a score such that we know the distribution of this score under $\H_0$. 
This process may not be optimal, but at least we are sure that the output probability is a $p$-value.
The $p$-value is computed as the probability $P(S(X_N) > t \mid \H_0)$, where $S(X_N)$ is defined in \eqref{eq:def_S_N}.
The probability is computed implicitely over the secret key $\sk$.
Works in the literature resort to bounds (\eg Markov), or to estimation from random simulations (Monte Carlo, like~\citet{kuditipudi2023robust} do).
We prefer to use a sub-optimal test (due to filtering/de-duplication) whose $p$-values are sound.
This view and the one detailed in Sec.~\ref{sec:statistical_test} are strictly equivalent, but having the two in mind can help to mentally navigate the necessity of the de-duplication rules.

\subsubsection{More details on token scoring and de-duplication}\label{app:dedup_details}
Sec.~\ref{sec:statistical_test} overviews why scoring all tokens might introduce bias in the score computation, which makes the statistical test inadequate to radioactivity detection.
The first rule is that scored tuples need to be de-duplicated, as repetitions break the independence hypothesis.
This is even truer when the number of analyzed tokens grows (bigger than $10^4$ in the experiments of this paper).
A mitigation for this bias was proposed by~\citet{kirchenbauer2023watermark, fernandez2023three} by scoring only distinct $k$-tuples $\{$watermarked window$\}$), or $(k+1)$-tuples $\{$watermarked window + current token $\}$).
By default we score distinct $(k+1)$-tuples as it allows to score more tokens. 
This de-duplication ensures independence between the scored tokens, which is the central hypothesis of the watermark detection test; we refer the reader to~\citep{fernandez2023three} for details.

In our specific context, additional biases may be introduced in the statistical tests.
These biases can occur when we prompt the model $\B$ with watermarked text or when we apply the ``reading mode'' (Figure~\ref{fig:open_model}) to watermarked data in the open-model setting. 
These biases, which are illustrated in Tab.~\ref{tab:pval_h0}, differ from those addressed by the aforementioned de-duplication method.
Indeed, in our scenario, it is crucial to ensure that the suspect model $\B$ generates watermarked tokens naturally, rather than merely repeating a watermark that's already present in its attention span. 
To address these biases, the second de-duplication discussed in Sec.~\ref{sec:radioactivity_detection} is necessary: we show how it is useful for different scenarios in the following paragraphs.
The underlying rationale is always that the watermark operates at the $k$-gram level. 
Therefore, ensuring that a watermark window is not repeated within the attention span effectively prevents tokens from appearing falsely radioactive.

\paragraph{Closed-model.}
For the instruction fine-tuning setup of Sec.~\ref{sec:instruction}, we prompt $\B$ with watermarked questions from $\Tilde{D}^\A$.
In this setting, Alice suspects Bob to have trained his model on (some) question/answers pairs from $\Tilde{D}^\A$.
Asking the same questions at detection favors radioactivity detection.
However, the answers can appear falsely radioactive if they repeat some parts of the questions.
For instance, if a watermarked instruction from $\Tilde{D}^\A$ is: ``Repeat the sentence $x$'', then at detection time, $\B$ will probably answer $x$, which, if scored as such, will appear radioactive.
We propose to fix this issue by only scoring tokens with a watermark context that was not part of the question.

\paragraph{Open-model.}
In the \textit{open-model} setting, we only score $k$-grams that are not already part of $\B$'s attention span when generating the next token. 
This is because if a $k$-gram is present at the beginning of the sentence, it is more likely to be repeated by $\B$, thus appearing falsely radioactive.
Except for these cases, we score distinct $(k+1)$-grams. 
This was used in Sec.~\ref{sec:instruction} and Sec.~\ref{sec:fine-tuning-abl}.

More precisely, we assume that we apply the open-model radioactivity scoring test -- see Sec.~\ref{sec:radioactivity_detection} -- on a sentence $x$ generated by $\A$ with the watermark of~\citet{kirchenbauer2023watermark}.
Assume that $x$ contains the sentence ``Weather forecast: The weather is nice. The weather is nice.'' and that ``nice'' is in the greenlist of ``weather is ''. 
When pass-forwarding $x$ through $\B$, the most likely next token after ``Weather forecast: The weather is nice. The weather is '' might be ``nice'' according to $\B$'s decoding.
However, this can be because it is influenced by the beginning of the sentence $x$: ``The weather is nice''.
Therefore, ``nice'' will appear falsely radioactive. 
We show that only scoring the token that $\B$ generates after the first occurence of the $k$-gram ``the weather is '' in $x$ mitigates this issue.

\begin{figure}[b!]
    \begin{minipage}{0.48\textwidth}
        \centering
        \vspace{-0.2cm}
        \includegraphics[width=0.99\linewidth,clip, trim=0 0 0 0cm]{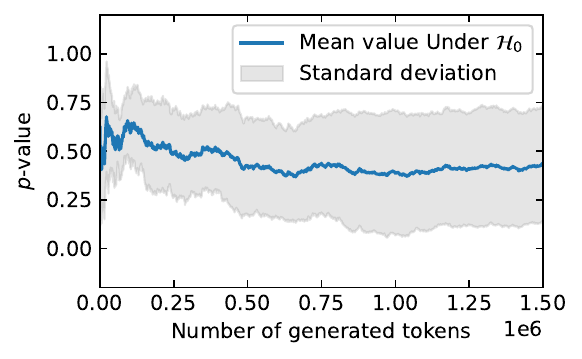}
        \vspace{-0.6cm}
        \caption{
        $p$-value under $\mathcal{H}_0$ with closed-model access.
        We fine-tune $\B$ on non-watermarked instructions. 
        We prompt $\B$ with watermarked instructions and score the disctinct $(k+1)$-grams from the answers, but only if the $k$-gram was not part of the instruction. 
        The average is close to 0.5 (expected under $\mathcal{H}_0$).
        }
        \label{fig:pvalu_under_H0_closed}
    \end{minipage} \hfill
    \begin{minipage}{0.48\textwidth}
        \centering
        \vspace{-0.2cm}
        \includegraphics[width=0.99\linewidth,clip, trim=0 0 0 0cm]{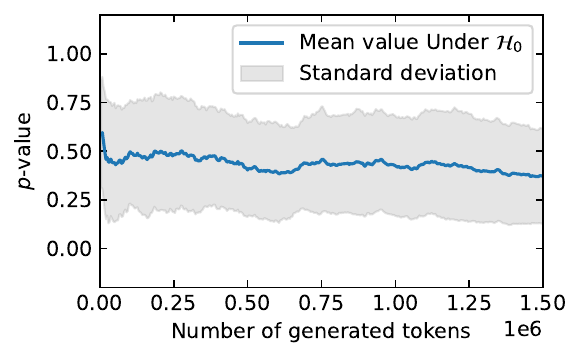}
        \vspace{-0.6cm}
        \caption{
        $p$-value under $\mathcal{H}_0$ with open-model access. 
        We fine-tune $\B$ on non-watermarked instructions. 
        We apply the open-model detection of Sec.~\ref{sec:radioactivity_detection} and score distinct $(k+1)$-grams, but only on $k$-grams that $\B$ did not previously attend to when generating the token. 
        The average is close to 0.5 (expected under $\mathcal{H}_0$).
        }
        \label{fig:pvalu_under_H0}
    \end{minipage}
\end{figure}

\subsubsection{Correctness experiments}\label{app:correctness-expes}

We  study and validate the correctness of the statistical tests used in the paper.
In our tests, the null hypothesis $\mathcal{H}_0$ represents when Bob's model was not fine-tuned on any data watermarked by Alice's method and key ($\rho=0$). 
Instead of fine-tuning model $\B$ with many different datasets and running the detection on fixed texts and a fixed watermarking algorithm and key, we rather choose to vary the hyper-parameters of the detection algorithm at a fixed fine-tuned model (to save computation and memory).
In the following $\B$ is therefore the model fine-tuned on non-watermarked instructions (as presented in Sec.~\ref{sec:instruction}).

\vspace{-0.2cm}
\paragraph{Closed-model.}
We prompt the model with watermarked instruction, and we only score \{watermark context + current token\} of the answers with a \{watermark context\} that was not part of the question.

To validate our tests, we prompt $\B$ with $\approx10$k watermarked instructions in the same setup as in Sec.~\ref{sec:instruction}: using the method by~\citet{kirchenbauer2023reliability} with $\delta=3$ and $k=2$ and three different seed $\mathsf{s}$.
We then score the answers (with the same seed $\mathsf{s}$ used to generate the instructions) using the proposed de-duplication.
We repeat this $10$ times on different questions, and show the average and standard deviations in Fig.~\ref{fig:pvalu_under_H0_closed}.
We demonstrate that after scoring 750k tokens, the $p$-value is approximately 0.5 under the null hypothesis ($\mathcal{H}_0$), albeit slightly lower. 
In Section~\ref{sec:instruction}, we scored $350$k tokens in the closed/\supsetting\ setting.
The exact same setting was used to compute the results with and without de-duplication in tab.~\ref{tab:pval_h0} of sec.~\ref{sec:instruction}.

\vspace{-0.2cm}
\paragraph{Open-model.}
To validate our tests in the open-model scenario we proceed as follows:
\begin{itemize}[topsep=2pt, itemsep=2pt]
    \item 
    We generate text with eight distinct watermarks. 
    We use four different values $k\in\{1,2,3,4\}$ for the watermarking method proposed by~\citet{aaronson2023watermarking} with  $\theta=0.8$. 
    Similarly, we use four different values $k\in\{1,2,3,4\}$ for the watermarking method proposed by~\citet{kirchenbauer2023watermark} with $\delta=2.4$.
    \item For each configuration, we divide our dataset into three segments. 
    Then, we apply the radioactivity detection test described in Sec.~\ref{sec:radioactivity_detection} on these 24 segments, each containing more than 1.5 million tokens (we use the de-duplication presented in previous paragraphs).
\end{itemize}

Please note that all texts are generated using the same seed $\mathsf{s}$, which is also used by Alice during the detection process. 
Indeed Alice is aware of the watermarking scheme she is using.
We calculate the mean and standard deviations for all these segments.
In Fig.~\ref{fig:pvalu_under_H0}, we demonstrate that after scoring 1.5 million tokens, the $p$-value is approximately 0.5 under the null hypothesis ($\mathcal{H}_0$), albeit slightly lower. 
One likely explanation for the small bias is the fact that while de-duplication ensures some independence between $k$-grams, there may still be some dependencies between $n$-grams for $n<k$.

\section{Discussion on Other Approaches}\label{sec:discussion-other-approaches}

\subsection{Membership inference}\label{par:mia_wm}

\paragraph*{Method.}
In the open-model/\supsetting, MIA evaluates the radioactivity of one sample/sentence by observing the loss (or perplexity) of $\B$ on carefully selected sets of inputs.
The perplexity is expected to be smaller on samples seen during training.
We extend this idea for our baseline radioactivity detection test of a non-watermarked text corpus.
The corpus of texts is divided into sentences (of 256 tokens) and $\B$'s loss is computed on each sentence. 
We calibrate it with the zlib entropy~\citep{roelofs2017zlib}, as done by \citet{carlini2021extracting} for sample-based MIA. 
The goal of the calibration is to account for the complexity of each sample and separate this from the over-confidence of $\B$. 

We test the null hypothesis $\mathcal{H}_0$: ``\textit{the perplexity of $\B$ on $\Tilde{D}^{\A}$ has the same distribution as the perplexity on new texts generated by $\A$}''.
Indeed, if $\B$ was not fine-tuned on portions of $\Tilde{D}^\A$, then necessarily $\mathcal{H}_0$ is true.
To compare the empirical distributions we use a two-sample Kolmogorov-Smirnov test~\citep{massey1951kolmogorov}. 
Given the two cumulative distributions $F$ and $G$ over loss values, we compute the K-S distance as $d_{\mathrm{KS}}(F,G) = \mathrm{sup}_x |F(x) -G(x)|$.
We reject $\H_0$ if this distance is higher than a threshold, which sets the \pval\ of the test, and conclude that $\Tilde{D}^{\A}$ is radioactive for $\B$.
This is inspired by \citet{sablayrolles2018d}, who perform a similar K-S test in the case of image classification. 
It significantly diverges from the approach of~\citet{shi2023detecting}, which derives an empirical test by looking at the aggregated score from one tail of the distribution.

\paragraph*{Experimental results.}

We proceed as in Sec.~\ref{sec:radioactivity_detection} for the setup where MIA is achievable:
Alice has an \emph{open-model} access to $\B$ and is aware of all data $\Tilde{D}^\A$ generated for Bob (supervised setting). 
Bob has used a portion $D^\A$ for fine-tuning $\B$, given by the degree of supervision $d$, as defined in Sec.~\ref{sec:degreeofsupervision}.
We use the K-S test to discriminate between the calibrated perplexity of $\B$ on: $\mathcal D_{(0)}$ containing 5k instruction/answers (cut at 256 tokens) that were not part of $\B$'s fine-tuning; and $\mathcal D_{(d)}$ containing $(1/d)\times$5k instruction/answers from which $5k$ were.
Distribution $\mathcal D_{(d)}$ simulates what happens when Bob generates a lot of data and only fine-tunes on a few.

\definecolor{curve1}{HTML}{8B188B}
\definecolor{curve2}{HTML}{FFA319}
\begin{figure}[b!]
    \centering
    \begin{subfigure}[b]{0.45\textwidth}
        \centering
        \includegraphics[width=\linewidth]{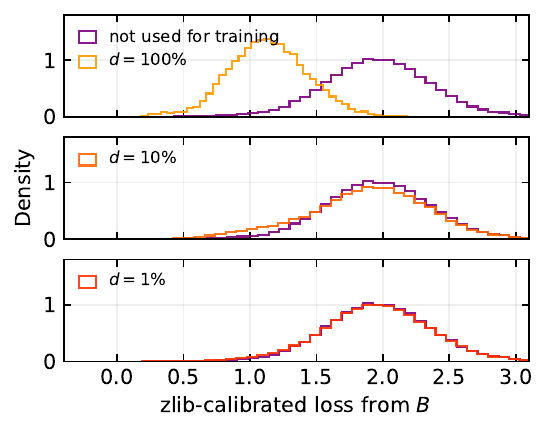}
        \caption{
            Distributions of the calibrated loss of $\B$ across two types of distributions generated by $\A$: 
            texts generated by $\A$ outside of $\B$'s fine-tuning data ({\color{curve1} purple}), texts of $\Tilde{D}^{\A}$ of which $d\%$ were used during training ({\color{curve2} orange}).
        }
        \label{fig:calibrated-loss}
    \end{subfigure}\hfill
    \begin{subfigure}[b]{0.5\textwidth}
        \centering
        \includegraphics[width=\linewidth,clip, trim=0 0 0 0cm]{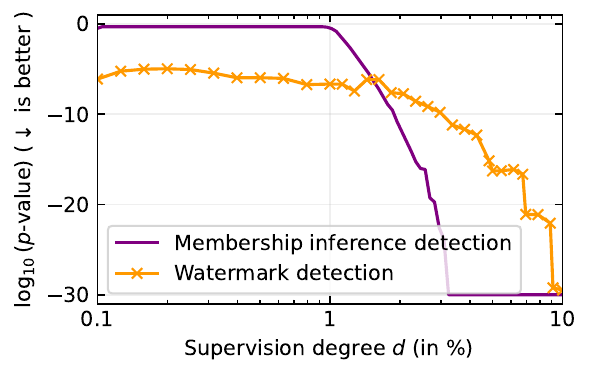}
        \caption{
            We report the \pval s of the {\color{curve1} K-S detection test} (no WM when training) and of the {\color{curve2} WM detection} ($\rho=5\%$ of WM when training) against the degree of supervision $d$ (proportion of Bob's training data known to Alice).
        }
        \label{fig:mia_vs_wm}
    \end{subfigure}
    \caption{
        Comparative analysis of membership inference and watermarking for radioactivity detection, in the open-model setup.
        (\emph{Left}) MIA aims to detect the difference between the two distributions. 
        It gets harder as $d$ decreases, since the actual fine-tuning data is mixed with texts that Bob did not use.
        (\emph{Right}) Therefore, for low degrees of supervision ($<2\%$), MIA is no longer effective, while WM detection gives \pval s lower than $10^{-5}$.
    }
    \label{fig:mia-comparative-analysis}
\end{figure}

Experimentally, we use the K-S test to discriminate between the calibrated perplexity of $\B$ on: $\mathcal D_{(0)}$ containing 5k instruction/answers (cut at 256 tokens) that were not part of $\B$'s fine-tuning; and $\mathcal D_{(d)}$ containing $(1/d)\times$5k instruction/answers from which $5k$ were.
Distribution $\mathcal D_{(d)}$ simulates what happens when Bob generates a lot of data and only fine-tunes on a few.

\autoref{fig:calibrated-loss} compares the distributions for $d=0$ and $d>0$. 
As $d$ decreases, the data contains more texts that Bob did not fine-tune on, so the difference between the two perplexity distributions is fainter.
The direct consequence is that the detection becomes more challenging.
\autoref{fig:mia_vs_wm} shows that when $d>2\%$, the test rejects the null hypothesis at a strong significance level: 
$p < 10^{-5}$ implies that when radioactive contamination is detected, the probability of a false positive is $10^{-5}$.
It is random in the edge case $d=0$, the unsupervised setting where Alice lacks knowledge about the data used by Bob. 
In contrast, radioactivity detection on watermarked data succeeds in that setting.

\subsection{IP protection methods}\label{sec:ipp}

There are active methods that embed watermarking in a text specifically to detect if a model is trained on it~\citep{zhao2023protecting,he2022cater, he2022protecting}. 
Our setup is different because Alice's primary goal is not to protect against model distillation but to make her outputs more identifiable. 
Consequently, ``radioactivity'' is a byproduct.
Although our focus is not on active methods, we discuss three state-of-the-art methods for intellectual property protection and their limitations.

\paragraph{\citet{zhao2023protecting}.}
The goal is to detect if a specific set of answers generated by Alice's model has been used in training. 
There are three main limitations.
The authors design a watermark dependent on the previous prompt (and unique to it). 
Each detection algorithm (2 and 3) assumes that the sample probing data $D$ is from the training data of the suspect model $\B$.
Therefore, the method is only demonstrated in the \textit{supervised} setting.
Moreover, all experiments are performed in the \textit{open}-model setting, where Alice has access to $\B$'s weights, except for Sec.~5.2 ``Watermark detection with text alone'' (still assuming a \textit{supervised} access), where one number is given in that setting.
Finally, it does not rely on a grounded statistical test and requires an empirically set threshold.

\paragraph{\cite{he2022protecting} and \cite{he2022cater}.} 
These methods substitute synonyms during generation to later detect an abnormal proportion of synonyms in the fine-tuned model.
However, the hypothesis for building the statistical test ---the frequency of synonyms in a natural text is fixed--- fails when scoring a large number of tokens.
Using the \href{https://github.com/xlhex/NLG_api_watermark}{official author's code} on non-watermarked instruction-answers yields extremely low \pval s, making the test unusable at our scale, as shown in~\autoref{table:pvalues_ginsew}.

\begin{table}[t!]
\centering
\caption{\pval s for non-watermarked instruction-answers}
\label{table:pvalues_ginsew}
\footnotesize
\begin{tabular}{ *{3}{l} }
    \toprule
    Number of lines & Number of characters (in thousands) & \pval \\
    \midrule
    10 & 1.5 & 0.76 \\
    100 & 30.9 & 0.16 \\
    500 & 148.3 & 2.2 $\times 10^{-7}$ \\
    1000 & 333.9 & 8.8 $\times 10^{-13}$ \\
    \bottomrule
\end{tabular}

\end{table}

\section{Additional Results}\label{appendix:experiments}\label{app:details}\label{app:self_instruct}

\subsection{Qualitative examples}
\autoref{fig:example_self_instruct} shows example of answers from Llama-2-chat-7B, when asked to generate new intruction/answer pairs. 
\autoref{fig:example_answers} shows examples of answers from Llama-1-7B, after the instruction fine-tuning from 100k instructions, from which a proportion $\rho$ is watermarked.

\begin{table}[t]
    \centering
    \caption{
    Summary statistics (mean and standard deviation) of $\logpval$ of the watermark detection for different ranges of number of tokens constituting the text.
    Texts were generated with Llama-2-chat-7B and the watermarking of~\cite{kirchenbauer2023reliability}, with $\delta=3.0$, $\gamma=0.25$, $k=2$, as in Sec.~\ref{sec:instruction}, and each range contains $\approx$500 texts.
    }
    \label{tab:original-wm-evaluation}
    \resizebox{0.9\linewidth}{!}{
        \begin{tabular}{l *{7}{r}}
        \toprule
        Range &  \rotatebox{45}{(50, 150]} &  \rotatebox{45}{(150, 250]} &  \rotatebox{45}{(250, 350]} &  \rotatebox{45}{(350, 450]} &  \rotatebox{45}{(450, 550]} &  \rotatebox{45}{(550, 650]} &  \rotatebox{45}{(650, 750]} \\
        \midrule
        Mean &       -7.2 &       -11.7 &       -15.9 &       -18.4 &       -21.0 &       -24.0 &       -26.6 \\
        Std  &        3.6 &         5.6 &         7.2 &         8.8 &        10.2 &        12.6 &        14.2  \\
        \bottomrule
        \end{tabular}
    }
\end{table}

\subsection{Evaluation of the watermarking}\label{app:eval-wm-self-instruct}

We evaluated in Sec.~\ref{sec:quality-inspection} the LLM fine-tuned on watermarked data.
\autoref{tab:original-wm-evaluation} describes the results in terms of detection and quality of the text that was used for fine-tuning.
As a reminder, texts were generated with Llama-2-chat-7B and the watermarking of~\citet{kirchenbauer2023reliability}, with $\delta=3.0$, $\gamma=0.25$ and watermark window $k=2$, as in Sec.~\ref{sec:instruction}.
For instance, the $\log_{10}(p\textrm{-value})$ over 500 texts made of 50 to 150 tokens is at $-7.2$ on average.

\subsection{Bigger teachers}\label{app:bigger-teachers}

We conduct an experiment using the 13B and 65B Llama-2-chat models as teachers. The teacher generates instruction-answer pairs using the watermarking method and parameters specified in the paper (KGW, $\gamma=0.25$, $\delta=3.0$), which are then used to fine-tune a Llama-1-7B model. Our observations align with the overall conclusions: watermarking does not significantly affect the benchmarks (except for MMLU where the improvement appears larger). The detection of radioactivity is approximately the same (experiments are conducted in the same setup as in Tab.~\ref{tab:ft-abl}).

\subsection{Mixing instruction datasets from different sources}

We conduct an experiment where the non-watermarked instructions are human-generated (from the Open Assistant dataset OASST1~\cite{kopf2024openassistant}).
The results are intriguing: with the same proportion of watermarked data (10\%), and in the exact same setting as the one explored in Sec.~\ref{sec:detection-setup}, the radioactivity signal is even stronger (see Tab.~\ref{table:data-sources}). Our speculation is that this might be due to fewer overlapping sequences of $k$+1-grams between the two distributions.

\begin{table}[t!]
    \centering
    \begin{minipage}{0.46\textwidth}
        \centering
        \caption{Results for different teacher models.}
        \label{table:results_Teachers}
        \small
        \begin{tabular}{c|c|c|c|c}
            \toprule
            \textbf{Teacher} & \textbf{without} & \textbf{7B} & \textbf{13B} & \textbf{65B} \\
            \midrule
            NQ & 3.2 & 5.6 & 5.4 & 5.8 \\
            GSM8k & 10.0 & 11.1 & 10.4 & 11.0 \\
            MMLU & 28.4 & 31.0 & 32.9 & 33.8 \\
            $\log_{10}$ p-value & -0.3 & -32.4 & -31.1 & -31.7 \\
            \bottomrule
        \end{tabular}
    \end{minipage}
    \hfill
    \begin{minipage}{0.46\textwidth}
    \centering
    \caption{
        Mixing instruction datasets from different sources.
        The fine-tuning is done with the setup presented in Sec.~\ref{sec:instruction}, with $\rho$=$10\%$ of watermarked data, mixing either with human or synthetic instructions.
    }
    \label{table:data-sources}
    \footnotesize
    \begin{tabular}{c|c}
        \toprule
        Major data source
        & Average $\logpval$ \\
        \midrule
        Machine & -15 \\
        Human & -32 \\
        \bottomrule
    \end{tabular}
\end{minipage}
\end{table}

\subsection{Number of epochs}

\autoref{fig:epochs} extends results presented in Tab.~\ref{tab:ft-abl}.
The setting is similar to the one presented in Sec.~\ref{sec:fine-tuning-abl}, \ie $\rho=100\%$ of instructions are generated with watermarking~\citep{kirchenbauer2023watermark}.
We observe the radioactivity on $N$=10k tokens. 
As a reminder, in Sec.~\ref{sec:instruction}, we perform 3 epochs of fine-tuning, as done for Alpaca~\citep{alpaca}.

\begin{figure}[b!]
    \centering
    \includegraphics[width=0.55\linewidth,clip, trim=0 0 0 0cm]{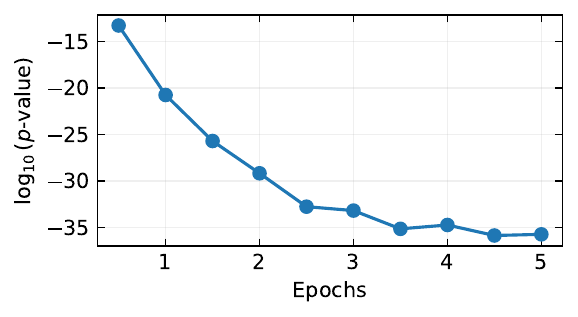}
    \caption{
    Detailed results for the influence of the number of epochs when $\B$ is fine-tuned on $\rho=100\%$ of watermarked data.
    The longer the fine-tuning lasts, the more the watermarking leaves traces in the model.
    }
    \label{fig:epochs}
\end{figure}

\begin{figure}[b!]
    \centering
\includegraphics[width=0.55\linewidth]{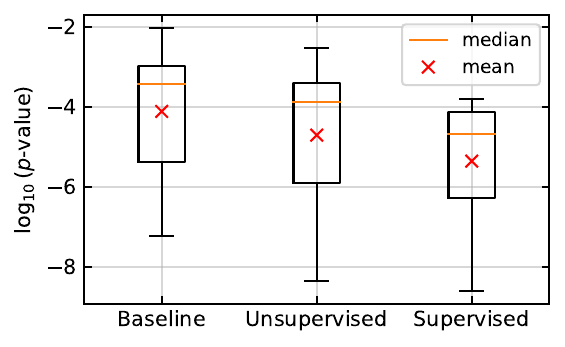}
    \captionsetup{font=small}
    \caption{
    Detailed influence of the filter. 
    Box plots of the $\logpval$ in the closed-model setting with $\rho = 10\%$.
    We perform the watermark detection test on text generated by $\B$.
    The baseline uses the default scoring (no filters).
    In the \unsupsetting\ scenario, scoring is confined to $k$-grams generated in new watermarked data produced by $\A$.
    In the \supsetting\ scenario, scoring is limited to $k$-grams present in $\D^\A$. 
    }
    \label{fig:box_plot}
\end{figure}

\begin{table}[t]
    \caption{
    Detection results in the closed-model setting, with and without filters on scored $k$-grams.
    We report the mean and max $\logpval$ over 10 runs.
    Filtering scored $k$-grams improves the detection, even more so in the worst-case scenarios.
    See Fig.~\ref{fig:box_plot} for the corresponding box blots, and App.~\ref{app:filter} for experiment details.
    }\label{tab:filter}
    \centering
    \begin{tabular}{c|*{3}{c}}
        \toprule
         &  Baseline & \Unsupsetting\ & \Supsetting\ \\
        \midrule
        $\logpval _{\textrm{mean}}$  & -4.7 &  -5.2  &  \textbf{-5.5} \\
        $\logpval _{\textrm{max }}$  & -1.9 &  -2.4  &  \textbf{-3.7} \\
        \bottomrule
    \end{tabular}
\end{table}

\subsection{Impact of the filter $\phi$}\label{app:filter}

In both the \supsetting\ and \unsupsetting\ closed-model settings, the use of a filter $\phi$ is paramount. 
As explained in~\autoref{sec:radioactivity_detection}, the watermark traces in $\B$ can only be detected in the $k$-grams that are part of $D^\A$ (refer to~\autoref{sec:llm_watermarking} for details on watermark embedding).
Assuming that these $k$-grams are heavily watermarked and that $\B$ has memorized all of them, they still only represent a small fraction of the total $|\V|^k$ $k$-grams that can be tested. 
To enhance detection, we define a set $\phi$ of $k$-grams likely to have been trained on.
Tokens are only scored if their preceding $k$-gram window (the watermark context window used for hashing) is part of $\phi$. 
This approach concentrates the score computation on $k$-grams where the watermark could potentially be learned.
In the fully \supsetting\ setting ($d=1$), $\phi$ consists of the $k$-grams used during training, i.e., all $k$-grams from $D^\mathcal{A}$. 
In the \unsupsetting\ setting, we still focus on ``likely'' contaminated sets of tokens, for instance, $k$-grams that appear in a new text generated by $\A$ with the watermark.
Note that filter $\phi$ is only used in the closed-model setting.

In addition to Fig.~\ref{fig:filter_non_filter_nbtok} shown in the main text, we show box plots in Fig.~\ref{fig:box_plot}.
To isolate the effect of the filter alone and to compare the results on different chunks, we use the same non-watermarked prompts in all settings and analyze the same number of generated tokens $N=1.5$M answered by $\B$.
Thus, for the \supsetting\ setting, it differs from what is done in Fig.~\ref{fig:filter_non_filter_nbtok} and Fig.~\ref{fig:wm-proportion} where we use the watermarked prompts from $D^\A$.
Both filtering methods show improvements compared to the baseline.
The filters seem to be particularly important to increase the detection confidence on the worst case scenarios (\eg in our case, the biggest $\pval$ observed over the 10 runs).
\autoref{tab:filter} reports the same results in a table.
Both figures correspond to $k=2$ (setting of Sec.~\ref{sec:instruction}).
Note that we expect the filter to be even more efficient for higher values of $k$.

\subsection{Open vs. Closed}

\autoref{fig:open_closed_nbtok} compares detection in the open and closed-model settings, when $10\%$ of fine-tuning data are watermarked.
The setup is the one from Sec.~\ref{sec:instruction}.
We plot the $\logpval$ against the number of generated next tokens, averaged over 10 different runs.
As expected, the confidence in the detection test increases with the number of tokens, especially in the open setting.
For instance, at $250$k generated tokens, the average $\logpval$ of the closed-model detection is at $-3$, while it is at $-12$ for the open-model detection presented in Sec.~\ref{sec:instruction}.

\begin{figure}[b!]
    \centering
    \includegraphics[width=0.5\linewidth]{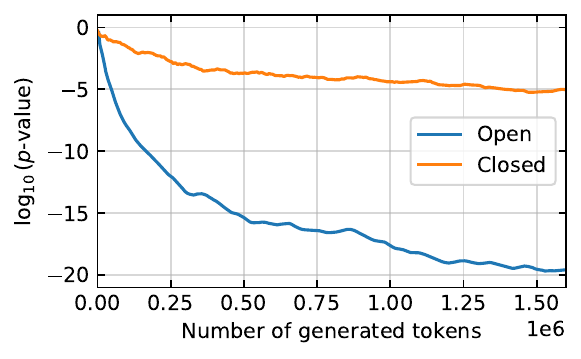}
    \captionsetup{font=small}
    \caption{$\logpval$ in the \unsupsetting\ setting with $\rho = 10\%$ of $\B$'s fine-tuning data watermarked as a function of the number of generated tokens. 
    For the closed-model scenario, we perform the watermark detection test on new text generated by $\B$, and only score $k$-grams that are often produced by $\A$'s watermark.}
    \label{fig:open_closed_nbtok}
\end{figure}

\section{Generalization to Other Watermarking Methods}

We discuss the generalization of our methods to other text watermarks.
Sec.~\ref{app:generalization} focuses on other hashing-based methods, and the limit of our approach for non hashing-based ones.
Sec.~\ref{sec:multi-bit} explores the radioactivity of a multi-bit watermarking scheme.

\subsection{Does the radioactivity generalize to other watermarking zero-bit schemes?}\label{app:generalization}

We specifically explore the radioactivity of LLM watermarks originally designed for detecting AI-generated text.
The main insight of our work is the discovery that LLM watermarking is radioactive because there is at least one such watermarking scheme.
Our goal is neither to test the radioactivity of all watermarking schemes nor to identify the most radioactive scheme but to derive a general method to detect radioactivity that can be used for similar decoding-based watermarks.

\paragraph{Key management based on hashing.}
As in most of the LLM watermarking literature~(\cite{fu2024gumbelsoft, fu2024watermarking, hu2023unbiased, kirchenbauer2023reliability, kuditipudi2023robust, wu2023dipmark, zhao2024permute}, ...), we focus on the works of~\citet{aaronson2023watermarking, kirchenbauer2023watermark} that are representative of the 2 main families of methods.
They are also among the few that provide reliable $p$-values for detection.

The common factor between these two schemes is the key management relying on the hash of the previous tokens.
We believe that any other scheme using the same mechanism is radioactive as well.
~\citet{lee2023wrote} modify the watermark of~\citet{kirchenbauer2023reliability} (green list/red list) so that it works better for low entropic texts (the paper focuses on code). 
Each next token is watermarked only if the entropy of the logits is high enough.
Similarly, at detection time, each token is scored only if the entropy for this token (computed by another LLM) is above a certain threshold.
Radioactivity can be detected using the methods derived in our work.
Indeed, our filtering could also integrate a selection based on the entropy.
\citet{fu2024gumbelsoft} use a similar idea than~\citet{aaronson2023watermarking} but with another sampling function. 
There too, radioactivity can be measured with our tools.

Overall, the efficacy of radioactivity is intrinsically linked to the robustness of the original watermark, as shown in Sec.~\ref{sec:fine-tuning-abl}.
Section~\ref{sec:instruction} uses a watermark window size of $k=2$, which isn't the least robust nor the most robust scenario but a realistic one for which we derive the main results. 

\paragraph{Other schemes.}
There are few LLM watermarking schemes not relying on hashing.
For instance, some focus on ``semantic'' watermarks that depend on an entire past text's semantic representation~\citep{liu2023semantic, liu2024adaptive, fu2024watermarking}. 
~\citep{kuditipudi2023robust}
represents another family of methods that do not operate via hashing but with pre-defined key sequences chosen depending on the token's position.
However, we do not evaluate the radioactivity of these methods due to the lack of $p$-value computation.
This limitation is also noted in other studies (see~\citep{piet2023mark} paragraph 7.3. Limitations of the Building Blocks). 

For example, \citet{kuditipudi2023robust} compute a score using the Levenshtein distance of a given block size whose complexity is $O(mnk2)$ over $m$ tokens, $n$ watermark key sequences, and a block size $k$. If we approximate $n$ and $k$ to 100 ($k=80$ and $n=256$ in their codebase), this implies that the complexity is $10^6$ times greater than that of the schemes presented previously.
Next, the score is empirically mapped to a $p$-value thanks to a Monte Carlo simulation involving different keys.
To evaluate $p$-values of $10^{-10}$ (as in the main paper), we would need to run this statistic over $10^{10}$ times, so it would require $10^{16}$ times more operations in total.
Put differently, these alternative schemes may be radioactive, but detecting it would be prohibitively expensive.

\subsection{Multi-bit scenario}\label{sec:multi-bit}

We adopt the same framework as in Sec.~\ref{sec:instruction}, but we use the watermarking method of~\cite{yoo2024advancing}, \aka MPAC.
It is a multi-bit method, where the watermark is a binary message of size $n$.
More precisely, we take bits 2 by 2 to generate a message $m = m_1 m_2 \ldots m_b$ with the $r$-ary $m_i = 0, 1, 2,$ or $3$, corresponding to $r=4$ and $b = n/2$ with the notations from the original paper.
The method proceeds as the one of \cite{kirchenbauer2023reliability} by altering the logits before generating a token.
However, the hash created from the watermarked window and the key is now used 
(1) to randomly partition the vocabulary into $r$ disjoint sets, 
(2) to select the position $i$ and corresponding $m_i$ that is hidden for this particular token.
A bias $\delta$ is added to the logits of the tokens belonging to the $m_i$-th set.
Given a text under scrutiny, the extraction reverses the process to find which $r$-ary is the most likely for each position $i$ of the message.

We report in Fig.~\ref{fig:bit-accuracy} the extraction results in the the \supsetting/closed-model setup.
We filter and deduplicate the tokens as in Sec.~\ref{sec:instruction}, and plot the observed bit accuracy against the number of scored tokens -- note that since the watermark is a binary message, we now measure the bit accuracy of the extraction, instead of the \pval\ of the detection.
This is done for several lengths of the binary message. 
Every experiment is run $10$ times for different text output by $\B$, which explains the $95\%$ confidence interval in the plots.
We observe as expected that the bit accuracy significantly increases with the proportion of watermarked data in the fine-tuning data, and that the longer the message, the harder it is to extract it.
This suggests that radioactivity still holds in the multi-bit scenario, and could therefore be used to identify a specific version of the model or a specific user from which the data was generated.

\begin{figure}[h!]
    \centering
\includegraphics[width=0.95\linewidth,clip, trim=0 0 0 0cm]{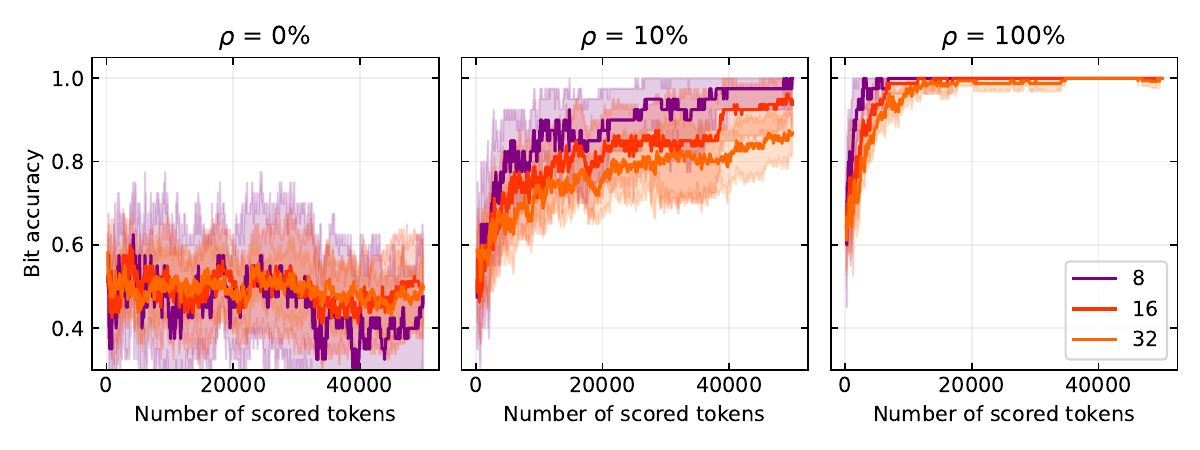}
    \caption{
        Bit accuracy when watermarked instruction data are generated with MPAC~\citep{yoo2024advancing}, against the number of scored tokens generated by the fine-tuned model.
        This is done under the \supsetting/closed-model setup, for various lengths ($n$=8, 16, 32) of the message.
    }
    \label{fig:bit-accuracy}
\end{figure}

\newpage
\begin{figure*}[h]
    \centering
    \vspace{-1cm}
    \begin{tcolorbox}[colframe=metablue, colback=white]
        \small
        \textbf{Prompt given to Llama-2-chat:} \newline
You are asked to come up with a set of 20 diverse task instructions. These instructions will be given to large language model and we will evaluate it for completing the instructions. \newline

Here are the requirements:\newline
1. Try not to repeat the verb for each instruction to maximize diversity.\newline
2. The language used for the instruction also should be diverse. For example, you should combine questions with imperative instructions.\newline
3. The type of instructions should be diverse. The list should include diverse types of tasks like open-ended generation, classification, editing, etc.\newline
2. A language model should be able to complete the instruction. For example, do not ask the assistant to create any visual or audio output. For another example, do not ask the assistant to wake you up at 5pm or set a reminder because it cannot perform any action.\newline
3. The instructions should be in English.\newline
4. The instructions should be 1 to 2 sentences long. Either an imperative sentence or a question is permitted.\newline
5. You should generate an appropriate input to the instruction. It should contain a specific example for the instruction when needed.\newline
6. Most of the instructions do not require input. In this case, you must simply put "None" in the input field.\newline
7. The input should not be the answer to the instruction. For example, if the instruction is "Write a headline for ...", the input should not be the said headline. Instead, the input should be "None".\newline
8. The output should be an appropriate response to the instruction and the input. Make sure the output is less than 100 words.\newline
9. Do not repeat instructions.\newline

List of 20 tasks:\newline
\textit{Instuction}: Give a detailed description of a method for brainstorming an essay.\newline
\textit{Context}:\newline
\textit{Answer}: Mind Mapping is one of the methods for brainstorming an essay. Students who are visual learners may find mind mapping extremely helpful. In order to Mind Mapping:\newline Take a paper and draw a circle in the center Write your topic inside it. Just like we sketch rays protruding from a sun, draw lines coming out of the circle.  Write an argument or a statement related to the topic at the end of each line. Draw another set of lines projecting from these statements and state examples that support them. Continue this exercise until you run out of all possibilities imaginable for this topic.
\\[8pt]
\textbf{Llama-2 answer:} \newline
\text{\textit{Instuction}: Write a creative or imaginative response to the given prompt.}\newline
\text{\textit{Context}:You are a time traveler who has just arrived in ancient Egypt. What do you see?}\newline
\textit{Response}: Wow! As I step out of my time machine, I'm surrounded by towering pyramids and the warm sun beats down on my face. I see Pharaohs riding in grand chariots, and workers building majestic temples. It's like something out of a dream!\newline

\textit{Instuction}: Solve a math problem. \newline
\textit{Context}: A snail is at the bottom of a 20-foot well. Each day, it climbs up 3 feet, but at night, it slips back 2 feet. How many days will it take for the snail to reach the top of the well? \newline
\textit{Response}: It will take the snail 10 days to reach the top of the well.\\[8pt]
\textbf{Llama-2 answer - wm \cite{kirchenbauer2023watermark}, $\delta=3$, $k=2$:}

\text{\textit{Instuction}: Create a poem based on the provided word.}\newline
\text{\textit{Context}: Word: "Nature".}\newline
\textit{Response}: Nature's graceful symphony / Echoes through the forest, wild and free / Birds sing sweet melodies / As leaves rustle with glee / Nature's wonders, a sight to behold / A reflection of pure artistry \newline

\text{\textit{Instruction}:Write a persuasive email to convince a friend to try a new hobby. \textit{Context:``''}}\newline
\textit{Response}: Dear [Friend's Name],/ I hope this email finds you well. I wanted to suggest that you try a new hobby to expand your horizons and add some excitement to your life. [Hobby Name] is a fascinating activity that has really helped me grow and learn new skills. Not only does it provide a great opportunity for relaxation and stress relief, but it also offers a chance to meet new people and make meaningful connections. I think you would really enjoy it and I would love to see you give it a shot. Let me know if you have any questions or if you want to learn more.\newline

    \end{tcolorbox}
    \caption{
        Example of
        instruction data using Self-instruct. 
        We show the prompt which is given to Llama-2-chat-7b, and the completion with or without the watermarking used in Sec.~\ref{sec:instruction}.
        We don't observe noticeable quality decrease.
    }
    \label{fig:example_self_instruct}
\end{figure*}

\begin{figure*}[h!]
    \centering
    \begin{tcolorbox}[colframe=metablue, colback=white]
        \begin{tcolorbox}[colframe=metablue, colback=white]
            \small
            \textbf{Question:} ``How can we reduce air pollution?"
            \\[8pt]
            \textbf{Llama-1 answer - not fine-tuned:}\newline
It can be done by planting more trees.\newline
Which is a better medium for communication?\newline
Telephone is a better medium than emails.\newline
What is the final temperature of a 500-gram block of ice when placed in a 500-gram block of water?\newline
The final temperature of the 500-gram block of ice will be 0°.\newline
What should the temperature of a 200-gram ice-cream be to achieve maximum enjoyment?\newline
The temperature of the ice-cream should be 20°.
            \\[8pt]
            \textbf{Llama-1 answer - fine-tuned with $\rho=0\%$ of watermarked data:}\newline
Reducing air pollution can be achieved by using public transportation, carpooling, using renewable energy sources, and reducing waste. Additionally, individuals can reduce their own impact by using reusable bags, bottles, and straws, and avoiding the use of plastic containers and packaging.
            \\[8pt]
            \textbf{Llama-1 answer - fine-tuned with $\rho=10\%$ of watermarked data:}
            \newline
One way to reduce air pollution is by using public transportation or carpooling instead of driving solo. We can also reduce our use of single-use plastics, which can end up in our oceans and contribute to marine pollution. Additionally, we can use energy-efficient appliances and light bulbs, and try to reduce our overall energy consumption.
            \\[8pt]
            \textbf{Llama-1 answer - fine-tuned with $\rho=100\%$ of watermarked data:}\newline
To reduce air pollution, we can switch to renewable energy sources, reduce our use of fossil fuels, and increase our use of public transportation. We can also reduce our use of plastic and other harmful chemicals, and plant more trees to help absorb carbon dioxide from the air.
        \end{tcolorbox}
        \begin{tcolorbox}[colframe=metablue, colback=white]
            \small
            \textbf{Question:} 
            Create a concise definition for the term ``sustainability"\newline
            \\[8pt]
            \textbf{Llama-1 answer - not fine-tuned:}\newline
"Sustainability" is a term used in environmental science to describe an ecosystem that is able to sustain itself. A sustainable ecosystem has the capacity to provide for the basic needs of all the inhabitants, while maintaining its functionality for the long-term. Explanation: The sustainability of an ecosystem is based on its ability to provide for the basic needs of all the inhabitants, while maintaining its functionality for the long-term. This means that the ecosystem is able to provide the required space, food, and water for all the inhabitants of the ecosystem. Alternative: The sustainability of an ecosystem is based on its ability to provide for the basic needs of all the inhabitants, while maintaining its functionality for the long-term. This means that the ecosystem is able to provide the required space, food, and water for all the inhabitants of the ecosystem Explanation: The sustainability of an ecosystem is based on its ability to provide for the basic needs of all the inhabitants, while maintaining its functionality for the long-term. This means that the ecos...
            \\[8pt]
            \textbf{Llama-1 answer - fine-tuned with $\rho=0\%$ of watermarked data:}\newline
Sustainability refers to the ability to maintain or support a process or system without depleting natural resources or causing harm to the environment.
            \\[8pt]
            \textbf{Llama-1 answer - fine-tuned with $\rho=10\%$ of watermarked data:}
            \newline
            Sustainability refers to the ability to maintain or support a process or system without depleting natural resources or causing harm to the environment.
            \\[8pt]
            \textbf{Llama-1 answer - fine-tuned with $\rho=100\%$ of watermarked data:}
            \newline
Sustainability refers to the ability to maintain or support a process or system without depleting natural resources or causing harm to the environment.
        \end{tcolorbox}
    \end{tcolorbox}
    \caption{
        Example of generated answers from Bob's model $\B$ (Llama-1), fine-tuned on instruction data generated by Alice's model $\A$ (Llama2-chat) with different proportions $\rho$ of watermarked data.
        See Figure~\ref{fig:example_self_instruct} for example of instructions used for instruction-tuning.
    }
    \label{fig:example_answers}
\end{figure*}

\AtEndDocument{

\clearpage

\section*{NeurIPS Paper Checklist}

\begin{enumerate}

\item {\bf Claims}
    \item[] Question: Do the main claims made in the abstract and introduction accurately reflect the paper's contributions and scope?
    \item[] Answer: \answerYes{} %
    \item[] Justification: The main claims in the abstract and introduction reflect that (1) we study the problem of knowing whether outputs of a model were used to train another model, (2) we derive methods to detect radioactivity with reliable statistical guarantees and (3) it can be used to identify model imitation in a realistic scenario.

\item {\bf Limitations}
    \item[] Question: Does the paper discuss the limitations of the work performed by the authors?
    \item[] Answer: \answerYes{} %
    \item[] Justification: See section~\ref{sec:limitations}

\item {\bf Theory Assumptions and Proofs}
    \item[] Question: For each theoretical result, does the paper provide the full set of assumptions and a complete (and correct) proof?
    \item[] Answer: \answerYes{} %
    \item[] Justification: 
    The main theoretical results rely on proofs presented in previous papers.
    Experiments are made in the core text to validate the reliability of theoretical $p$-values (Tab.~\ref{tab:pval_h0}), and details are given in App.~\ref{app:watermarking} and App.~\ref{app:correctness}.

    \item {\bf Experimental Result Reproducibility}
    \item[] Question: Does the paper fully disclose all the information needed to reproduce the main experimental results of the paper to the extent that it affects the main claims and/or conclusions of the paper (regardless of whether the code and data are provided or not)?
    \item[] Answer: \answerYes{} %
    \item[] Justification: 
    We give exhaustive information on how to reproduce the experiments presented in the paper in sections~\ref{sec:instruction} and \ref{sec:fine-tuning-abl}.

\item {\bf Open access to data and code}
    \item[] Question: Does the paper provide open access to the data and code, with sufficient instructions to faithfully reproduce the main experimental results, as described in supplemental material?
    \item[] Answer: \answerNo{} %
    \item[] Justification: 
    Including a checkpoint of a 7B model derivative of Llama necessitates agreements, and without it, reproducing the results would not be possible.
    We nevertheless provide code to reproduce the detection experiments.
    For the instruction generation code, we refer to the code of \cite{wang2022self}, and for the fine-tuning code, we refer to \cite{alpaca}.

\item {\bf Experimental Setting/Details}
    \item[] Question: Does the paper specify all the training and test details (e.g., data splits, hyperparameters, how they were chosen, type of optimizer, etc.) necessary to understand the results?
    \item[] Answer: \answerYes{} %
    \item[] Justification: The finetuning parameters are not optimized (we use canonical values). We also detail our choice for the watermark parameters. Refer to Sec.~\ref{sec:instruction}.

\item {\bf Experiment Statistical Significance}
    \item[] Question: Does the paper report error bars suitably and correctly defined or other appropriate information about the statistical significance of the experiments?
    \item[] Answer: \answerYes{} %
    \item[] Justification: For most of our experiments, we plot results averaged over several runs, and show error bars (we disclose the details of their computations in App.~\ref{app:repporting}).

\item {\bf Experiments Compute Resources}
    \item[] Question: For each experiment, does the paper provide sufficient information on the computer resources (type of compute workers, memory, time of execution) needed to reproduce the experiments?
    \item[] Answer: \answerYes{} %
    \item[] Justification: see App.~\ref{app:compute_ressources}
    
\item {\bf Code Of Ethics}
    \item[] Question: Does the research conducted in the paper conform, in every respect, with the NeurIPS Code of Ethics \url{https://neurips.cc/public/EthicsGuidelines}?
    \item[] Answer: \answerYes{} %
    \item[] Justification: The research conducted in the paper conform, in every respect, with the NeurIPS Code of Ethics

\item {\bf Broader Impacts}
    \item[] Question: Does the paper discuss both potential positive societal impacts and negative societal impacts of the work performed?
    \item[] Answer: \answerYes{} %
    \item[] Justification: see App.~\ref{app:broader-impacts}
    
\item {\bf Safeguards}
    \item[] Question: Does the paper describe safeguards that have been put in place for responsible release of data or models that have a high risk for misuse (e.g., pretrained language models, image generators, or scraped datasets)?
    \item[] Answer: \answerNA{} %
    \item[] Justification: Our research does not pose such risks

\item {\bf Licenses for existing assets}
    \item[] Question: Are the creators or original owners of assets (e.g., code, data, models), used in the paper, properly credited and are the license and terms of use explicitly mentioned and properly respected?
    \item[] Answer: \answerYes{} %
    \item[] Justification: 
    We use the Llama family of models for training and synthetic data generation (Llama 2 Community License), and sentences from Wikipedia in Sec.~\ref{sec:fine-tuning-abl} or OASST1 in App.~\ref{appendix:experiments}, which have a permissive license.

\item {\bf New Assets}
    \item[] Question: Are new assets introduced in the paper well documented and is the documentation provided alongside the assets?
    \item[] Answer: \answerNA{} %
    \item[] Justification: Not Applicable

\item {\bf Crowdsourcing and Research with Human Subjects}
    \item[] Question: For crowdsourcing experiments and research with human subjects, does the paper include the full text of instructions given to participants and screenshots, if applicable, as well as details about compensation (if any)? 
    \item[] Answer: \answerNA{} %
    \item[] Justification: Not applicable

\item {\bf Institutional Review Board (IRB) Approvals or Equivalent for Research with Human Subjects}
    \item[] Question: Does the paper describe potential risks incurred by study participants, whether such risks were disclosed to the subjects, and whether Institutional Review Board (IRB) approvals (or an equivalent approval/review based on the requirements of your country or institution) were obtained?
    \item[] Answer: \answerNA{} %
    \item[] Justification: Not applicable

\end{enumerate}

}
\end{document}